\documentclass{lmcs} 


\usepackage{hyperref,amsmath,amssymb}
\theoremstyle{plain} 

\newcommand{\Min}{\mathrm{Min}}
\newcommand{\Max}{\mathrm{Max}}

\newcommand{\source}{\mathsf{source}}
\newcommand{\target}{\mathsf{target}}
\newcommand{\col}{\mathsf{col}}

\newcommand{\vphi}{\varphi}

\newcommand{\Cons}{\mathsf{Cons}}
\newcommand{\FMD}{\mathsf{FMD}}
\newcommand{\oneFMD}{\mathsf{1playerFMD}}

\begin{document}

\title[Instructions]{One-to-Two-Player Lifting for Mildly Growing Memory}
\titlecomment{An extended abstract of the paper is published in the proceedings of STACS 2022.}
\thanks{}	

\author[A.~Kozachinskiy]{Alexander Kozachinskiy}[a]
\address{}	
\email{}  






\begin{abstract}
  \noindent We investigate a phenomenon of ``one-to-two-player lifting'' in infinite-duration two-player games on graphs with zero-sum objectives. More specifically, let $\mathcal{C}$ be a class of strategies. It turns out that in many cases, to show that all two-player games on graphs with a given payoff function are determined in $\mathcal{C}$, it is sufficient to do so for one-player games.  That is, in many cases the determinacy in $\mathcal{C}$ can be ``lifted'' from one-player games to two-player games. Namely, Gimbert and Zielonka~(CONCUR 2005) have shown this for the class of positional strategies. Recently, Bouyer et al.~(CONCUR 2020) have extended this to the classes of arena-independent finite-memory strategies. Informally, these are finite-memory strategies that use the same way of storing memory in all game graphs.

In this paper, we put the lifting technique into the context of memory complexity. The memory complexity of a payoff function measures, how many states of memory we need to play optimally in game graphs with up to $n$ nodes, depending on $n$. We address the following question. Assume that we know the memory complexity of our payoff function in one-player games. Then what can be said about its memory complexity in two-player games? In particular, when is it finite?

In this paper, we answer this questions for strategies with ``chromatic'' memory. These are strategies that only accumulate sequences of colors of edges in their memory. We obtain the following results.
\begin{itemize}
\item Assume that the chromatic memory complexity in one-player games is sublinear in $n$ on some infinite subsequence. Then the chromatic memory complexity in two-player games is finite.

\item We provide an example in which \textbf{\emph{(a)}} the chromatic memory complexity in one-player games is linear in $n$;  \textbf{\emph{(b)}} the memory complexity in two-player games is infinite.
\end{itemize}
Thus, we obtain the exact barrier for the one-to-two-player lifting theorems in the setting of chromatic finite-memory strategies. Previous results only cover payoff functions with constant chromatic memory complexity.

\end{abstract}

\maketitle

\section{Introduction}

We study two-player infinite-duration games on graphs. These games are of interest in many areas of computer science, ranging from purely theoretical disciplines, such as decidability of logical theories~\cite{muchnik1992games,zielonka1998infinite}, to more practically-oriented ones, such as controller synthesis~\cite{gradel2002automata}. 

These games are played as follows. There is a finite directed graph with a token. We will call this graph \emph{arena}. Initially, the token is placed in one of the nodes of the arena. In each turn, one of the two players takes the token and moves it to some other node. A restriction is that there must be an edge to the new location of the token. For each node of the arena, it is fixed in advance which of the players is the one to move the token in this node. The game proceeds for infinitely many turns. The outcome of the game is decided by the resulting trajectory of the token (it forms an infinite path in the arena).

We restrict ourselves to \emph{zero-sum} games. Correspondingly, the players will be called Max and Min from now on. In a zero-sum game, objectives of the players are defined through a \emph{payoff function} -- a function of the form $\vphi\colon C^\omega\to\mathcal{W}$, where $C$ is a set of \emph{colors}, and $(\mathcal{W}, \le)$ is an arbitrary linearly ordered set. Next, we assume that arenas are edge-colored by elements of $C$. To compute the outcome of a play (which will be an element of $\mathcal{W}$), we take the trajectory of the token in this play, then consider the infinite sequence of colors $\gamma \in C^\omega$ written on the edges of the trajectory, and, finally, apply $\vphi$ to $\gamma$. The aim of Max is to maximize $\vphi(\gamma)$, while the aim of Min is to minimize it (with respect to the ordering of $\mathcal{W}$).

As usually, a pair of strategies of the players in which the first strategy is the best response to the second one, and vice versa, is called an \emph{equilibrium}. Next, a strategy which belongs to some equilibrium is called \emph{optimal}. Now, a payoff function is called \emph{determined} if in every arena there exists an equilibrium with respect to this payoff function.

We will study determinacy with respect to restricted classes of strategies. Namely, if $\mathcal{C}$ is a class of strategies, then we say that a payoff function is determined in $\mathcal{C}$ if the following holds: in every arena there is an equilibrium for this payoff function in which both strategies are from $\mathcal{C}$. The smaller is $\mathcal{C}$, the stronger is this requirement. 

\medskip

One of the main research directions in the area of games of graphs is \emph{strategy complexity}. Its goal, broadly speaking, is to find out, for a payoff function $\vphi$ of our interest, what is the ``simplest'' class of strategies $\mathcal{C}$ in which $\vphi$ is determined. This is highly relevant when our task is to actually implement in practice one of the optimal strategies for $\vphi$. For instance, this is the case when we want to produce a device whose performance is measured by $\vphi$. If this device is meant 
to act in the environment, then the execution of this device can be modeled as a game -- between the controller of the device and the environment. In this framework, the controller realizes one of the strategies in this game. Ideally, we want an optimal performance w.r.t.~$\vphi$ at the lowest cost (in terms of the resources we need to implement the controller). The lower is strategy complexity of $\vphi$, the easier is this task.

Classically, there are two classes of strategies that are often considered in this context. One is the class of \emph{positional} strategies and the other is the class of \emph{finite-memory} strategies.

Let us first consider positional strategies.
A strategy is positional if, for every node $v$ of the arena, it always makes the same move when the token is in $v$, no matter what was the path of the token to this node. Sometimes these strategies are called \emph{memory-less} -- they do not need to ``remember'' anything about the previous development of the game. For brevity, we call payoff functions that are determined in the class of positional strategies \emph{positionally determined}. Classical examples of games with positionally determined payoff functions are Parity Games, Mean-Payoff Games and Discounted Games~\cite{mostowski1991games,ehrenfeucht1979positional,zwick1996complexity}.

These games, especially Parity Games, had a tremendous impact on such areas as verification, model checking and program analysis~\cite{emerson1991tree,emerson1993model,baldan2019fixpoint}. However, say, in controller synthesis, it is often required to consider more complex games, namely, those for which positional strategies do not suffice. This brings us to a more general class of strategies -- the class of finite-memory strategies.

Unlike positional strategies, finite-memory strategies can store some information about the previous development of the game. The point is that during the whole play, which is infinitely long, the amount of this information should never exceed some constant.

The storage of information in finite-memory strategies is carried out by \emph{memory skeletons}. A memory skeleton $\mathcal{M}$ is a deterministic finite automaton whose input alphabet is the set of colors. Now, an $\mathcal{M}$-strategy is a strategy which, informally, stores information according to the memory skeleton $\mathcal{M}$. To understand how it works, imagine that during the game, each time the token is shifted along some edge, the color of this edge is fed to $\mathcal{M}$. Then, at every moment, the current state of $\mathcal{M}$ represents the current content of the memory. Correspondingly, the moves of an $\mathcal{M}$-strategy depend solely on the current state of $\mathcal{M}$ and the current node with the token.

A strategy is finite-memory if it is an $\mathcal{M}$-strategy for some memory skeleton $\mathcal{M}$.
For brevity, we call payoff functions that are determined in the class of finite-memory strategies \emph{finite-memory determined}.

\begin{rem}
Finite-memory strategies as defined above are sometimes called  ``chromatic''.  This is because one can consider a more general definition. Namely, one can allow memory skeletons to take the whole edge as an input, not only its color. However, as shown by Le Roux~\cite{le2020time}, determinacy in general finite-memory strategies is equivalent to determinacy in chromatic finite-memory strategies. In this paper, we work only with chromatic finite-memory strategies.
\end{rem}

\subsection{One-to-two-player lifting}

One of the techniques in the area of strategy complexity is called \emph{one-to-two-player lifting}. Our paper is devoted to this technique. It relies on the notion of \emph{one-player arenas}.
An arena is called one-player if for one of the players the following holds: all the nodes of the arena from which this player is the one to move have exactly one out-going edge. This means that one of the players is given no choice and has only one way of playing. Correspondingly, there are two types of one-player arenas -- those in which Max has no choice and those in which Min has no choice.

It turns out that to study determinacy in some class of strategies $\mathcal{C}$, it is sometimes sufficient to consider only one-player arenas. As was shown by Gimbert and Zielonka~\cite{gimbert2005games}, this applies to the class of positional strategies. More specifically, their result states the following. Assume that a payoff function is such that all \emph{one-player} arenas have an equilibrium of two positional strategies\footnote{Note that in one-player arenas, one of the players has just one strategy (and this strategy is positional). So this requirement means that the other player has a positional strategy which is at least as good against the unique strategy of the opponent as any other strategy.} with respect to this payoff function. Then \emph{all} arenas, not only one-player ones, have an equilibrium of two positional strategies with respect to this payoff function. That is, then this payoff function is positionally determined. In a way, this means the positional determinacy of one-player games can always be ``lifted'' to two-player games. 

This result has fundamental significance for studying the positional determinacy.  This is because often one-player arenas are considerably easier to analyze than two-player ones. Indeed, assume we have an arena in which, say, Min has no choice. A question of whether such an arena has a positional equilibrium reduces to the following question. Is there a ``lasso'' (a simple path to a simple cycle over which we rotate infinitely many times)  which maximizes our payoff function over all infinite paths? Often this can be figured out with a simple graph reasoning. For instance, this is fairly easy for Parity Games and Mean Payoff Games. Thus, through the lifting theorem of Gimbert and Zielonka one gets simple proofs of positional determinacy of these games. In turn, proofs that existed prior to the paper of Gimbert and Zielonka were highly non-trivial.

Given such a success in the case of positional strategies, it is temping to extend this to larger classes of strategies. This was recently investigated for the class of finite-memory strategies by Bouyer et al.~in~\cite{bouyer2020games}. It turns out that the situation is quite different for this class. More specifically, Bouyer et al.~have constructed a payoff function such that \textbf{\emph{(a)}} all one-player arenas have an equilibrium of two finite-memory strategies with respect to this payoff function \textbf{\emph{(b)}} there is an arena (in fact, with just 2 nodes) which is not one-player and which has no equilibrium of two finite-memory strategies  with respect to this payoff function. 

Thus, the class of positional strategies admits one-to-two-player lifting and the class of finite-memory strategies does not. Bouyer et al.~suggested to study intermediate classes. Namely, their approach was as follows.
By definition, the class of finite-memory strategies
 is the union of the classes of $\mathcal{M}$-strategies over all memory skeletons $\mathcal{M}$. Let us now fix a memory skeleton $\mathcal{M}$ and consider the class of $\mathcal{M}$-strategies for this specific $\mathcal{M}$. Bouyer et al.~show that for every $\mathcal{M}$ this class admits one-to-two-player lifting.

More precisely, the lifting theorem of Bouyer et al.~states that for any memory skeleton $\mathcal{M}$ the following holds. Assume that a payoff function is such that all one-player arenas have an equilibrium of two $\mathcal{M}$-strategies. Then the same holds for all arenas, with exactly this memory skeleton $\mathcal{M}$. That is, then this payoff function is determined in $\mathcal{M}$-strategies.
 
Observe that positional strategies are exactly $\mathcal{M}$-strategies if the memory skeleton $\mathcal{M}$ has just one state.
Thus, the lifting theorem Bouyer et al. includes the lifting theorem of Gimbert and Zielonka as a special case.

Bouyer et al.~call payoff functions to which one can apply their lifting theorem 
\emph{arena-independent finite-memory determined}. That is, a payoff function $\vphi$ is arena-independent finite-memory determined if there exists a memory skeleton $\mathcal{M}$ such that $\vphi$ is determined in $\mathcal{M}$-strategies.

In the literature there is a number of games with arena-independent finite-memory determined payoff functions. For example, one can list games with $\omega$-regular winning conditions~\cite{buchi1969solving} and bounded multidimensional energy games~\cite{bouyer2008infinite}. In turn, unbounded multidimensional energy games are finite-memory determined but not arena-independently~\cite{chatterjee2014strategy}.

\subsection{Our results}

The aim of this work is to extend the lifting technique beyond the class of arena-independent finite-memory determined payoff functions. 

For payoff functions beyond this class, there is no single memory $\mathcal{M}$ skeleton which suffices for all arenas (here  ``suffices'' means the existence of an equilibrium of two $\mathcal{M}$-strategies). Instead, larger arenas require larger memory skeletons. This motivates a notion of the \emph{memory complexity} of a payoff function. It can be defined as follows. For every $n$ consider the minimal memory skeleton which is sufficient for all arenas with up to $n$ nodes (w.r.t.~our payoff function).  Let the size of this memory skeleton (that is, the number of its states) be $S_n$. Then we call the function $n\mapsto S_n$ the memory complexity of our payoff function. Observe that arena-independent finite-memory determined payoff functions have memory complexity $O(1)$.

The memory complexity is the decisive factor in practice -- if it grows too quickly, we might have no resources to implement optimal strategies for our payoff function. This complexity measure was studied for a number of payoff functions in~\cite{chatterjee2012energy,chatterjee2014strategy}

We initiate the study of the memory complexity in the context of one-to-two-player lifting. More specifically, we address the following question. Assume that we know the memory complexity of our payoff function in one-player arenas. Then what can be said about its memory complexity in all arenas? Thus,
our approach differs from the approach of Bouyer et al.~in the following regard. Instead of lifting determinacy in some fixed class of strategies from one-player arenas to all arenas, we lift bounds on the memory complexity.

 To formulate our results, we introduce the following notation.  Let $\mathbb{Z}^+$ denote the set of positive integers, and let $f\colon\mathbb{Z}^+\to\mathbb{Z}^+$ be a function. Then by $\FMD(f)$ we denote the class of all payoff functions $\vphi$ such that for all $n\in\mathbb{Z}^+$ there exists a memory skeleton $\mathcal{M}$ with at most $f(n)$ states such that every arena with at most $n$ nodes has an equilibrium of two $\mathcal{M}$-strategies with respect to $\vphi$. In other words, $\FMD(f)$ is the class of all payoff function with memory complexity at most $f$. We also introduce similar notation for one-player arenas. Namely, we let $\oneFMD(f)$ be the class of all payoff functions $\vphi$ such that for all $n\in\mathbb{Z}^+$ there exists a memory skeleton $\mathcal{M}$ with at most $f(n)$ states such that every \emph{one-player} arena with at most $n$ nodes has an equilibrium of two $\mathcal{M}$-strategies with respect to $\vphi$. Again, $\oneFMD(f)$ is the class of payoff functions whose memory complexity in one-player arenas is at most $f$.
Obviously, $\FMD(f) \subseteq \oneFMD(f)$. Additionally, we let $\FMD$ stand for the class of all finite-memory determined payoffs. Finally, let $\oneFMD$ be the class of all payoff functions $\vphi$ such that every one-player arena has an equilibrium of two finite-memory strategies w.r.t.~$\vphi$.

In this notation, the question we address in this paper can be formulated as follows: for which functions $f$ and $g$ do we have $\oneFMD(f) \subseteq \FMD(g)$? 

\begin{rem}
One could consider an alternative definition of $\FMD(f)$, in which different arenas of size up to $n$ may be mapped to different memory skeletons of size $f(n)$. Unfortunately, it is not clear how to extend results of this paper to this setting.
\end{rem}

\medskip

Before presenting our results, let us express previous ones in this notation.
For technical convenience, we assume from now on that the set $C$ of colors is finite. This is not an essential restriction, as any arena involves only finitely many colors. Hence, if $C$ is infinite, one can study, separately all finite subsets $C^\prime\subseteq C$, arenas that involve colors only from $C^\prime$.

First, let us understand what payoff functions are included\footnote{Here, formally, by $\FMD(1)$ we mean $\FMD(f)$ for the function $f\colon\mathbb{Z}^+\to\mathbb{Z}^+$ such that $f(n) = 1$ for all $n\in\mathbb{Z}^+$. More generally, if there is some expression in $n$ defining a function $f\colon\mathbb{Z}^+\to\mathbb{Z}^+$, we will use $\FMD$ of this expression instead of $\FMD(f)$. For example, if $f(n) = 2n^2 + 2$ for all $n\in\mathbb{Z}^+$, then we will write $\FMD(2n^2 + 2)$ instead of $\FMD(f)$.}~in $\FMD(1)$. By definition, these are payoff functions such that for every $n$ there is a memory skeleton $\mathcal{M}$ with 1 state such that all arenas with up to $n$ nodes are determined in $\mathcal{M}$-strategies -- or, equivalently, in positional strategies. Thus, $\FMD(1)$ is exactly the class of positionally determined payoff functions. Observe then that the lifting theorem of Gimbert and Zielonka can be stated as the equality $\oneFMD(1) = \FMD(1)$.

In fact, the lifting theorem of Bouyer et al.~asserts that, more generally, for any constant $k\in\mathbb{Z}^+$ we have $\oneFMD(k) = \FMD(k)$. Indeed, take any $\vphi\in\oneFMD(k)$. Our goal is to show that $\vphi\in\FMD(k)$. By definition, for every $n$ there exists a memory skeleton $\mathcal{M}$ with at most $k$ states such that all one-player arenas with at most $n$ nodes have an equilibrium of two $\mathcal{M}$-strategies w.r.t.~$\vphi$. A problem is that these $\mathcal{M}$ may be different for different $n$. However, since the set $C$ of colors is finite, there are only finitely many memory skeletons with up to $k$ states. One of them works for infinitely many $n$ -- and, hence, for all one-player arenas. Due to the lifting theorem of Bouyer et al., the same memory skeleton works for all arenas. Thus, since this memory skeleton has at most $k$ states, we have $\vphi\in\FMD(k)$.

Let us note that the class of arena-independent finite-memory determined payoffs is the class $\FMD(O(1)) = \bigcup_{k\in\mathbb{Z}^+}\FMD(k)$.

Finally, since lifting does not hold for the whole class of finite-memory strategies, we have $\oneFMD \neq \FMD$. In fact, this means that for some function $f$ we have $\oneFMD(f) \nsubseteq \FMD$. This is because
\[\FMD = \bigcup\limits_f \FMD(f),\qquad \oneFMD = \bigcup\limits_f\oneFMD(f)\]
over all $f\colon\mathbb{Z}^+\to\mathbb{Z}^+$. Why is it so? For example, let us show this for $\FMD$. We have to show that for any $\vphi\in\FMD$ and for every $n$ there exists a memory skeleton $\mathcal{M}$ such that all arenas with up to $n$ nodes have an equilibrium of two $\mathcal{M}$-strategies (w.r.t $\vphi$). A point is that, since $C$ is finite, for every $n$ the number of such arenas is also finite (w.l.o.g.~we may assume that between each pair of nodes there are at most $|C|$ edges). In each of these arenas, fix a pair of finite-memory strategies forming an equilibrium (this is possible since $\vphi\in\FMD$). This gives a finite set of finite-memory strategies such that every arena with up to $n$ nodes is determined in strategies from this set. It remains to set $\mathcal{M}$ to be the product of the memory skeletons of these strategies. Then all these strategies will be $\mathcal{M}$-strategies.

\medskip

We proceed to our main result. Let $\Omega(n)$ denote the set of functions $f\colon\mathbb{Z}^+\to\mathbb{Z}^+$ for which there exists $C > 0$ such that $f(n) \ge C n$ for all $n\in\mathbb{Z}^+$. We obtain the following lifting theorem:

\begin{thm}
\label{main_theorem}
 Consider any function $f\colon\mathbb{Z}^+\to\mathbb{Z}^+, f\notin\Omega(n)$. Define
$g\colon\mathbb{Z}^+\to\mathbb{Z}^+, g(n) = f\left(\min\left\{m\mid \frac{f(m)}{m + 1} \le \frac{1}{2n}\right\}\right)$.
Then $\oneFMD(f) \subseteq\FMD(g)$.
\end{thm}
First, why is the function $g$ well-defined? Since $f\notin \Omega(n)$, the fraction $f(m)/m$ gets arbitrarily close to $0$ for some $m$. Hence, the minimum in the definition of $g$ is always over a non-empty set.

Now consider the case when, as in the lifting theorem of Bouyer et al., the function $f$ is constant, that is $f(n) = k$ for some constant $k\in\mathbb{Z}^+$ and for all $n\in\mathbb{Z}^+$. Then we have $g(n) = k$  for all $n\in\mathbb{Z}^+$ as well. That is, our main results implies the equality $\oneFMD(k) = \FMD(k)$, and this equality is the lifting theorem of Bouyer et al.

It is instructive to consider an example when $f\notin\Omega(n)$ and is super-constant. Say, assume that $f(n) = O(n^\gamma)$ for some $\gamma < 1$. It is easy to see that then $g(n) = O(n^{\gamma/(1 - \gamma)})$. Now there is a gap between memory complexity in one-player arenas and in all arenas. The closer $\gamma$ is to $1$, the larger is this gap.

When $\gamma$ gets equal to $1$, Theorem \ref{main_theorem} becomes inapplicable. We demonstrate that this is not due to the weakness of our technique.

\begin{thm}
\label{thm:sharp}
 $\oneFMD(2n + 2) \nsubseteq \FMD$.
\end{thm}

This result shows the sharpness of Theorem \ref{main_theorem}.
Namely, in order to obtain at least some bound on the memory complexity in all arenas, the memory complexity in one-player arenas should be a function not from $\Omega(n)$. In other words, it should be sublinear on some infinite subsequence. In turn, when it is already just linear, we might have no finite-memory determinacy.

Thus, our paper pushes the technique of one-to-two-player lifting to its limit. Unfortunately, this limit turns out to be very low. We are not aware of a payoff function which has been considered in the literature and to which one can apply Theorem \ref{main_theorem}, but which is not arena-independent finite-memory determined. For example, let us consider unbounded multidimensional games -- as we have indicated, they are finite-memory determined but not arena-independently. As shown in~\cite{jurdzinski2015fixed}, these games are in $\FMD(n^{O(1)})$. Here the constant in $O(1)$ depends on the dimension and the maximum of the norms of the weights.  In any case, this bound is not sufficient for Theorem \ref{main_theorem}.

Still, we provide an example of a payoff function to which our lifting theorem is applicable and the lifting theorem of Bouyer et al.~is not.
\begin{thm} 
\label{thm:example}
There exists a function $f\colon\mathbb{Z}^+\to\mathbb{Z}^+$ with $f\notin\Omega(n)$ and a payoff function from $\oneFMD(f)$ which is not arena-independent finite-memory determined.
\end{thm}

\subsection{Other related works and concluding remarks}

First, the exact analogs of the theorems of Gimbert and Zielonka and Bouyer et al.~for \emph{stochastic} games were obtained in other works of these authors~\cite{gimbert2016pure,bouyer2021arena}. We find it plausible that our result can be lifted to stochastic games as well. Le Roux and Pauly~\cite{le2018extending} obtained a \emph{two-to-many-players} lifting theorem. Namely, they show that, under some conditions, two-player finite-memory determinacy implies that all multiplayer games have finite-memory Nash equilibrium. A different approach to study finite-memory determinacy can be found in~\cite{le2018extending2}.

A natural open question is to extend lifting theorems to strategies with non-chromatic finite memory. As we mentioned, Le Roux~\cite{le2020time} has shown that non-chromatic finite memory can always be replaced by the chromatic one. Unfortunately, this transformation is rather costly -- the size of the memory grows exponentially in the number of nodes. So even the following modest question seems to be open: is there a payoff function which has constant non-chromatic memory complexity in one-player games but is not finite-memory determined in two-player games?

\bigskip

\textbf{Organization of the paper.} In Section \ref{sec:prel} we give preliminaries. In Section \ref{sec:tech} we give brief overviews of the proofs of our results. The proof of Theorem \ref{main_theorem} is given in Sections \ref{sec:special_case}--\ref{sec:main_proof}. Theorem \ref{thm:sharp} is proved in Section \ref{sec:sharp}. Theorem \ref{thm:example} is proved in Section \ref{sec:example}.

\section{Preliminaries}
\label{sec:prel}

\textbf{Notation.} We denote the set of positive integer numbers by $\mathbb{Z}^+$. Given a set $A$, by $A^*$ and $A^\omega$ we denote the sets of finite and, respectively, infinite sequences of elements of $A$. The length of a sequence $x\in A^*\cup A^\omega$ is denoted by $|x|$. We write $A = B\sqcup C$ for three sets $A, B, C$ if $A = B\cup C$ and $B\cap C = \varnothing$. Function composition is denoted by $\circ$.

\subsection{Arenas}

Following previous papers~\cite{gimbert2005games,gimbert2016pure,bouyer2020games,bouyer2021arena},  we call graphs on which our games are played \emph{arenas}. We start with some notation regarding arenas. First, take an arbitrary finite set $C$. We will refer to the elements of $C$ as \emph{colors}. Informally, an arena is just a directed graph with edges colored by elements of $C$ and with nodes partitioned into two sets.

\begin{defi}
\label{def:arena}
A tuple $\mathcal{A} = \langle V, V_\Max, V_\Min, E, \source, \target, \col\rangle$, where
\begin{itemize}
\item $V, V_\Max, V_\Min, E$ are four finite sets with $V = V_\Max\sqcup V_\Min$;
\item $\source, \target, \col$ are functions of the form $\source\colon E\to V, \target\colon E\to V, \col\colon E\to C$;
\end{itemize}
is called an \textbf{arena} if for every $v\in V$ there exists $e\in E$ with $v = \source(e)$.
\end{defi}
Elements of $V$ will be called \textbf{nodes} of $\mathcal{A}$ and elements of $E$ will be called \textbf{edges} of $\mathcal{A}$. We understand $e\in E$ as a directed edge from the node $\source(e)$ to the node $\target(e)$. There might be parallel edges and loops. Additionally, every edge $e$ of $\mathcal{A}$ is labeled by the color $\col(e)\in C$. Nodes from $V_\Max$ will be called nodes of Max and nodes from $V_\Min$ will be called nodes of Min. The out-degree of a node $v\in V$ is $|\{e\in E\mid \source(e) = v\}|$. By definition, every node in every arena has positive out-degree. An arena is called \textbf{one-player} if either all nodes of Max have out-degree $1$ or all nodes of Min have out-degree $1$.

Fix an arena $\mathcal{A} = \langle V, V_\Max, V_\Min, E, \source, \target, \col\rangle$. We extend the function $\col$ (which determines the coloring of the edges) to arbitrary sequences of edges by setting:
$\col(e_1 e_2 e_3\ldots) = \col(e_1) \col(e_2)\col(e_3)\ldots$ for $e_1, e_2, e_3,\ldots\in E$.

A non-empty sequence of edges $h = e_1 e_2 e_3 \ldots \in E^* \cup E^\omega$ is called a \textbf{path} if for every $1 \le n < |h|$ we have $\target(e_n) = \source(e_{n + 1})$. We define $\source(h) = \source(e_1)$. When $h$ is finite, we define $\target(h) = \target(e_{|h|})$. In addition, for every $v\in V$ we consider a $0$-length path $\lambda_v$ identified with $v$, for which we set $\source(\lambda_v) = \target(\lambda_v) = v$.  For every $v\in V$ we define $\col(\lambda_v)$ as the empty string.

\subsection{Infinite-duration games on arenas}
An arena $\mathcal{A} = \langle V, V_\Max, V_\Min, E, \source, \target, \col\rangle$ induces an infinite-duration two-player game in the following way. First, we call players of this game Max and Min. Informally, Max and Min interact by gradually constructing a longer and longer path in $\mathcal{A}$. In each turn one of the players extends a current path by some edge from its endpoint. Which of the two players is the one to move is determined by whether this endpoint belongs to $V_\Max$ or to $V_\Min$.

Formally, positions in the game are finite paths in $\mathcal{A}$. By definition, $\target(h)\in V_\Max$ for a finite path $h$ means that Max is the one to move in the position $h$; respectively, $\target(h)\in V_\Min$ means that Min is the one to move in the position $h$. A set of moves available in a position $h$ is the set $\{e\in E\mid \source(e) = \target(h)\}$.
 Making a move $e\in E$ in a position $h = e_1 e_2 \ldots e_{|h|}$ brings to a position $he = e_1 e_2 \ldots e_{|h|}e$.

We stress that no position is designated as the initial one. We assume that the game can start in any position of the form $\lambda_v, v\in V$, at our choice.

Next we proceed to a notion of strategies. Namely, a strategy of Max is a function 
\[\sigma\colon\{h\mid h\mbox{ is a finite path in $\mathcal{A}$ with } \target(h) \in V_\Max\} \to E\]
such that for every $h$ from the domain of $\sigma$ we have $\source(\sigma(h)) = \target(h)$. Respectively, a strategy of Min is a function
\[\tau\colon\{h\mid h\mbox{ is a finite path in $\mathcal{A}$ with } \target(h) \in V_\Min\} \to E\]
such that for every $h$ from the domain of $\tau$ we have $\source(\tau(h)) = \target(h)$.

Observe that if $\mathcal{A}$ is one-player, then one of the players has exactly one strategy. For technical consistency we assume that even when one of the players owns all the nodes of $\mathcal{A}$, the other player still has one ``empty'' strategy.

A strategy induces a set of positions \emph{consistent} with it (those that can be reached in a play against this strategy).
Formally, a finite path $h = e_1 e_2 \ldots e_{|h|}$ is consistent with a strategy $\sigma$ of Max if the following conditions hold:
\begin{itemize}
\item $\source(h)\in V_\Max \implies \sigma(\lambda_{\source(h)}) = e_1$;
\item for every $1\le i < |h|$ we have $\target(e_1 e_2\ldots e_i)\in V_\Max \implies \sigma(e_1 e_2\ldots e_i) = e_{i + 1}$.
\end{itemize}
Consistency with the strategies of Min is defined similarly. Further, the notion of consistency can be extended to infinite paths. Namely, given a strategy, an infinite path is consistent with it if all finite prefixes of this path are.

For $v\in V$ and for a strategy $\mathcal{S}$ of one of the players $\Cons(v, \mathcal{S})$ denotes the set of all finite and infinite paths that start at $v$ and are consistent with $\mathcal{S}$. For any strategy $\sigma$ of Max, strategy $\tau$ of Min and $v\in V$, there is a unique infinite path in the intersection $\Cons(v, \sigma)\cap \Cons(v,\tau)$. We denote this path by $h(v,\sigma,\tau)$ and call it \emph{the play} of $\sigma$ and $\tau$ from $v$.

\subsection{Payoff functions and equilibria}

We consider only zero-sum games; correspondingly, in our framework objectives of the players are always given by a \emph{payoff function}. A payoff function is any function of the form $\vphi\colon C^\omega \to \mathcal{W}$, where $(\mathcal{W}, \le)$ is a linearly ordered set. Informally, the aim of Max is to play in a way which maximizes the payoff function (with respect to the ordering of $\mathcal{W}$) while the aim of Min is the opposite one. Technically, to get the value of the payoff function on a play (which is an infinite path in the underlying arena) we first apply the function $\col$ to this play; this gives us an infinite sequence of colors; in conclusion, we apply $\vphi$ to the sequence of colors.

Any payoff function in a standard way induces a notion of an $\emph{equilibrium}$ of two strategies of the players (with respect to this payoff function). Let us first introduce a notion of an \emph{optimal response}. Namely, 
take a strategy $\sigma$ of Max and a strategy $\tau$ of Min. We say that $\sigma$ is a \textbf{uniformly optimal response} to $\tau$ if for all $v\in V$ and for all infinite $h\in \Cons(v,\tau)$ we have $\vphi\circ \col\big(h(v, \sigma,\tau)\big) \ge \vphi\circ\col(h)$. The inequality here, of course, is with respect to the ordering of $\mathcal{W}$.
Similarly, we call $\tau$ a \textbf{uniformly optimal response} to $\sigma$ if for all $v\in V$ and for all infinite $h\in \Cons(v,\sigma)$ we have $\vphi\circ \col\big(h(v, \sigma,\tau)\big) \le \vphi\circ\col(h)$.
Next, we call a pair $(\sigma, \tau)$ a \textbf{uniform equilibrium} if $\sigma$ and $\tau$ are  uniformly optimal responses to each other.
\begin{lem}
\label{cartesian}
 For any arena $\mathcal{A}$ and for any payoff function  $\vphi$, the set uniform equilibria in $\mathcal{A}$ w.r.t.~$\vphi$ is a Cartesian product. 
\end{lem}
\begin{proof}
See Subsection \ref{subsec:cartesian}.
\end{proof}

 Strategies which belong to some uniform equilibrium will be called \textbf{uniformly optimal}. 

\begin{rem}
Each payoff function induces a total preoder on $C^\omega$. Two payoff functions that induce the same preorder have the same set of equilibria. Due to this reason, previous papers in this line of work~\cite{gimbert2005games,gimbert2016pure,bouyer2020games,bouyer2021arena} do not consider payoff functions at all. Instead, they directly consider total preorders on $C^\omega$, to which they refer as \emph{preference relations}. We prefer to use a terminology of payoff functions, as it is more standard. Of course, this does not make our results less general -- any preference relation is induced by some payoff function. 
\end{rem}

\subsection{Positional strategies and finite-memory strategies}

\textbf{Positional strategies.} 
A strategy $\mathcal{S}$ of one of the players is called positional if for any two positions $h_1, h_2$ from its domain we have $\target(h_1) = \target(h_2) \implies \mathcal{S}(h_1) = \mathcal{S}(h_2)$. In other words,  $\mathcal{S}(h)$ depends solely on $\target(h)$. It makes convenient to consider positional strategies as functions on the set of nodes of the corresponding players (rather than on the set of the positions of this player). I.e., positional strategies of Max can be identified with functions of the form $\sigma \colon V_\Max\to E$ such that $\source(\sigma(v)) = v \mbox{ for all } v\in V_\Max$. Similarly, positional strategies of Min can be identified with functions of the form $\tau \colon V_\Min\to E$ such that $\source(\tau(v)) = v \mbox{ for all } v\in V_\Min$.

Let us fix some notation regarding positional strategies. First, every edge $e\in E$ is a path (of length $1$) and hence also a position in the game induced by $\mathcal{A}$. If $\mathcal{S}$ is a positional strategy of one of the players, we let $E_\mathcal{S}$ be the set of edges that are consistent with $\mathcal{S}$. Observe the following feature of positional strategies: the set of paths (positions) that are consistent with a positional strategy $\mathcal{S}$ is exactly the set of paths that consist only of edges from $E_\mathcal{S}$.

Given a positional strategy $\mathcal{S}$ of one of the players, by $\mathcal{A}_\mathcal{S}$ we denote the arena
\[
\mathcal{A}_\mathcal{S} = \langle V, V_\Max, V_\Min, E_\mathcal{S}, \source, \target, \col\rangle.\] That is, $\mathcal{A}_\mathcal{S}$ is obtained from $\mathcal{A}$ by deleting all edges that are inconsistent with $\mathcal{S}$. Observe that the arena $\mathcal{A}_\mathcal{S}$ is one-player; each node of the player who plays $\mathcal{S}$ has exactly one out-going edge in $\mathcal{A}_\mathcal{S}$.

Instead of saying ``an equilibrium of two positional strategies''  we will simply say ``a positional equilibrium''.

\medskip

\textbf{Finite-memory strategies.}  A memory skeleton is a deterministic finite automaton $\mathcal{M} = \langle M, m_{init}\in M, \delta\colon M\times C\to M\rangle$ whose input alphabet is the set $C$ of colors. Here $M$ is the set of states of $\mathcal{M}$, the state $m_{init}\in M$ is a designated initial state, and $\delta$ is the transition function of $\mathcal{M}$. By $|\mathcal{M}|$ we denote the number of states of a memory skeleton $\mathcal{M}$. Given $m\in M$, we extend $\delta(m,\cdot)$ to finite sequences of elements of $C$ in a standard way. Now, a strategy $\mathcal{S}$ of one of the players is called an $\mathcal{M}$-strategy if for any two positions $h_1$ and $h_2$ from the domain of $\mathcal{S}$ it holds that 
\[\big[\target(h_1) = \target(h_2) \mbox{ and }\delta(m_{init}, \col(h_1)) = \delta(m_{init}, \col(h_2)) \big] \implies \mathcal{S}(h_1) = \mathcal{S}(h_2).\]

In other words, $\mathcal{S}(h)$ depends solely on $\target(h)$ (the node with the token in the position $h$) and $\delta(m_{init}, \col(h))$ (the state into which $\mathcal{M}$ comes after reading the sequence of colors along $h$).

A strategy $\mathcal{S}$ of one of the players is called a finite-memory strategy if it is an $\mathcal{M}$-strategy for some memory skeleton $\mathcal{M}$.
Instead of saying ``an equilibrium of two finite-memory strategies'' or ``an equilibrium of two $\mathcal{M}$-strategies'' we will simply say ``a finite-memory equilibrium'' and ``an $\mathcal{M}$-strategy equilibrium''.

\subsection{Determinacy and memory complexity}

\begin{defi} Let $\mathcal{C}$ be a class of strategies. We say that a payoff function $\vphi$ is determined in $\mathcal{C}$ if every arena has a uniform equilibrium of two strategies from $\mathcal{C}$ w.r.t.~$\vphi$. In particular,
\begin{itemize}
\item if $\mathcal{C}$ is the class of positional strategies, then we call $\vphi$ \textbf{positionally determined}.

\item if $\mathcal{C}$ is the class of finite-memory strategies, then we call $\vphi$ \textbf{finite-memory determined}.

\item if $\mathcal{C}$ is the class of $\mathcal{M}$-strategies for some memory skeleton $\mathcal{M}$, then we call $\vphi$ \textbf{arena-independent finite-memory determined}.
\end{itemize}
\end{defi}

For our results it is important that we require equilibria to be uniform in these definitions. That is, it is important to have a single pair of strategies from $\mathcal{C}$ which is an equilibrium no matter in which node the game starts. As far as we know, this is the case for all positionally and finite-memory determined payoff functions that have been considered in the literature. 

Next we provide definitions regarding the memory complexity.
\begin{defi}

Let $\FMD$ denote the class of functions $\vphi\colon C^\omega\to\mathcal{W}$ such that $C$ is a finite set, $\mathcal{W}$ is linearly ordered and $\vphi$ is finite-memory determined.  Let $\oneFMD$ denote the class of functions $\vphi\colon C^\omega\to\mathcal{W}$ such that $C$ is a finite set, $\mathcal{W}$ is a linearly ordered set and such that the following holds: every one-player arena (with edges colored by elements of $C$) has a uniform finite-memory equilibrium w.r.t.~$\vphi$.

 Next,
consider any function $f\colon \mathbb{Z}^+\to\mathbb{Z}^+$. Let $\FMD(f)$ denote the class of functions $\vphi\colon C^\omega\to\mathcal{W}$ such that $C$ is a finite set, $\mathcal{W}$ is a linearly ordered set and such that the following holds: for all $n\in\mathbb{Z}^+$ there exists a memory skeleton $\mathcal{M}$ over the set $C$ with $|\mathcal{M}| \le f(n)$ such that all arenas (with edges colored by elements of $C$) with at most $n$ nodes have a uniform $\mathcal{M}$-strategy equilibrium w.r.t.~$\vphi$. Similarly, let $\oneFMD(f)$ denote the class of functions $\vphi\colon C^\omega\to\mathcal{W}$ such that $C$ is a finite set, $\mathcal{W}$ is a linearly ordered set and such that the following holds: for all $n\in\mathbb{Z}^+$ there exists a memory skeleton $\mathcal{M}$ over the set $C$ with $|\mathcal{M}| \le f(n)$ such that all one-player arenas (with edges colored by elements of $C$) with at most $n$ nodes have a uniform $\mathcal{M}$-strategy equilibrium w.r.t.~$\vphi$.

\end{defi}

\section{Overviews of the Proofs}
\label{sec:tech}

\subsection{Theorem \ref{main_theorem}}

First, let us give the exact statement of the lifting theorem of Bouyer et al.
\begin{thm}[\cite{bouyer2020games}]
\label{all_arenas_lifting}
For any payoff function $\vphi$ and for any memory skeleton $\mathcal{M}$ the following holds. 
 Assume that all one-player arenas have a uniform $\mathcal{M}$-strategy equilibrium w.r.t.~$\vphi$. Then all arenas have a uniform $\mathcal{M}$-strategy equilibrium  w.r.t.~$\vphi$.
\end{thm}

Our main technical contribution is the following strengthening of Theorem \ref{all_arenas_lifting}.

\begin{thm}
\label{lifting} For any payoff function $\vphi$ and for any $n\in\mathbb{Z}^+$ the following holds.
 Let $\mathcal{M}$ be a memory skeleton such that all one-player arenas with at most $2 n\cdot |\mathcal{M}| - 1$ nodes have a uniform $\mathcal{M}$-strategy equilibrium w.r.t.~$\vphi$. Then all arenas with at most $n$ nodes have a uniform $\mathcal{M}$-strategy equilibrium w.r.t.~$\vphi$.
\end{thm}
\begin{proof}[Derivation of Theorem \ref{main_theorem} from Theorem \ref{lifting}]
Take any $\vphi\in\oneFMD(f)$. Our goal is to show that $\vphi\in\FMD(g)$, where $g$ is as in Theorem \ref{main_theorem}. That is, our goal is to establish for every $n\in\mathbb{Z}^+$ a memory skeleton $\mathcal{M}$ with at most $g(n)$ states such that all arenas with at most $n$ nodes have a uniform $\mathcal{M}$-strategy equilibrium.

Take any $n\in\mathbb{Z}^+$. By definition, $g(n) = f(m)$ for some $m\in\mathbb{Z}$ such that $\frac{f(m)}{m + 1} \le \frac{1}{2n}$. Since $\vphi\in\oneFMD(f)$, there exists a memory skeleton $\mathcal{M}$ with at most $f(m)$ states such that all one-player arenas with at most $m$ nodes have a uniform $\mathcal{M}$-strategy equilibrium. Now, since $\frac{f(m)}{m + 1} \le \frac{1}{2n}$, we have $m\ge 2n\cdot f(m) - 1 \ge 2n \cdot |\mathcal{M}| - 1$. By Theorem \ref{lifting}, this means that all arenas with at most $n$ nodes have a uniform $\mathcal{M}$-strategy equilibrium. Since $\mathcal{M}$ has at most $f(m) = g(n)$ states, we are done.
\end{proof}

Before discussing our technique, 
let us briefly overview how  Bouyer et al.~establish Theorem \ref{all_arenas_lifting}. They start by defining ``$\mathcal{M}$-monotone payoff functions'' and ``$\mathcal{M}$-selective payoff functions''.  Then they show that any payoff function which is $\mathcal{M}$-monotone and $\mathcal{M}$-selective is determined in $\mathcal{M}$-strategies. Finally, they show that for any non-$\mathcal{M}$-monotone and for any non-$\mathcal{M}$-selective payoff function there exists a one-player arena which has no uniform $\mathcal{M}$-strategy equilibrium w.r.t.~this payoff function. This also gives a \emph{characterization} of $\mathcal{M}$-determinacy: a payoff function is determined in $\mathcal{M}$-strategies if and only if it is $\mathcal{M}$-monotone and $\mathcal{M}$-selective.

In this paper, we obtain Theorem \ref{lifting} (and, thus, Theorem \ref{all_arenas_lifting}) more directly. For the sake of simplicity, in Section \ref{sec:special_case} we prove it in a special case when $\mathcal{M}$ is a memory skeleton with just one state. In this special case, $\mathcal{M}$-strategies are positional strategies.

\begin{prop}[Special case of Theorem \ref{lifting}]
\label{positional_lifting} For any payoff function $\vphi$ and for any $N\in\mathbb{Z}^+$ the following holds. Assume that all one-player arenas with at most $2 N - 1$ nodes have a uniform positional equilibrium w.r.t.~$\vphi$. Then all arenas with at most $N$ nodes have a uniform positional equilibrium w.r.t.~$\vphi$.
\end{prop}

As all papers in this line of works, we build upon the inductive technique first invented by Gimbert and Zielonka~\cite{gimbert2005games}. Our contribution here is a more direct exposition of this technique, with the emphasis on the size of arenas.

\medskip

We extend Proposition \ref{positional_lifting} to all memory skeletons\footnote{Our technique in this part is rather similar to a technique from a recent paper of Bouyer et al.~\cite{bouyer2021arena} (see the arXiv version~\cite{arena} of their paper for more details). In this paper, they give a direct proof of an analogue of Theorem \ref{all_arenas_lifting} for stochastic games.} in two steps.  We first prove an analogue of Proposition \ref{positional_lifting} for so-called $\mathcal{M}$-trivial arenas. Informally, these are arenas where states of $\mathcal{M}$ are ``hardwired'' into nodes. In such arenas, $\mathcal{M}$-strategies degenerate to positional strategies. We show that Proposition \ref{positional_lifting} is true even when only $\mathcal{M}$-trivial arenas are taken into account (in the assumption and in the conclusion).

We then derive Theorem \ref{lifting} from this using the product arena construction~\cite[Chapter 2]{automata_toolbox}. Take any (two-player) arena $\mathcal{A}$ with up to $n$ nodes. We have to derive the existence of an $\mathcal{M}$-strategy equilibrium in $\mathcal{A}$ from the assumption of Theorem \ref{lifting}. It is a classical observation that $\mathcal{M}$-strategies in $\mathcal{A}$ can be viewed as positional strategies in the product arena $\mathcal{M}\times\mathcal{A}$. This product arena is obtained by first pairing states of $\mathcal{M}$ with nodes of $\mathcal{A}$, and then by drawing edges of $\mathcal{A}$ in all possible ways that are consistent with the transition function of $\mathcal{M}$. Now we only have to establish a positional equilibrium in $\mathcal{M}\times \mathcal{A}$. This arena is $\mathcal{M}$-trivial, so we use Proposition \ref{positional_lifting} for $\mathcal{M}$-trivial arenas and $N = n\cdot|\mathcal{M}|$. The size of $\mathcal{M}\times \mathcal{A}$ is the product of the sizes of $\mathcal{M}$ and $\mathcal{A}$, so it does not exceed $N$. It remains to show that all one-player $\mathcal{M}$-trivial arenas with up to $2N - 1 = 2n\cdot|\mathcal{M}| - 1$ nodes have a positional equilibrium. Indeed, by the assumption of Theorem \ref{lifting}, all one-player arenas (not only $\mathcal{M}$-trivial) of this size have an $\mathcal{M}$-strategy equilibrium. But in $\mathcal{M}$-trivial arenas these $\mathcal{M}$-strategy equilibria are automatically positional.

The full proof of Theorem \ref{lifting} is given in Appendix \ref{sec:main_proof}.

\subsection{Theorem \ref{thm:sharp}}
Let the set of colors be $C = \{-1, 1\}$. We define a payoff function  $\psi \colon C^\omega \to \{0, 1\}$ as follows. We set $\psi(c_1 c_2 c_3 \ldots) = 1$ if and only if either $\big(\lim_{n\to\infty} \sum_{i = 1}^n c_i = +\infty\big)$ or $\big(\sum_{i = 1}^n c_i = 0 \mbox{ for infinitely many } n\big)$.
We assume the standard ordering on $\{0, 1\} = \psi(C^\omega)$, so that $1$ is interpreted as victory of Max and $0$ is interpreted as victory of Min.

We show that $\psi\in\oneFMD(2n + 2)\setminus\FMD$.
In fact, this payoff function was defined by Bouyer et al.~in~\cite[Section 3.4]{bouyer2020games}. They have shown that this payoff function is finite-memory determined in one-player arenas but not in two-player arenas. So our contribution here is an upper bound $\psi\in\oneFMD(2n + 2)$ on its memory complexity in one-player arenas. In other words, for every $n$ we provide a memory skeleton $\mathcal{M}_n$ with $2n + 2$ states such that every one-player arena $\mathcal{A}$ with up to $n$ nodes has a uniform $\mathcal{M}_n$-strategy equilibrium. Let us describe the main ideas needed to obtain this upper bound. In this overview, we only consider those one-player arenas where all nodes of Min have out-degree 1. We use similar ideas for one-player arenas of the opposite type (but they require a bit more care).

 It will be more convenient to refer to the elements of $C$ as \emph{weights} rather than as colors. Correspondingly, by the weight of a path we will mean the sum of the weights of its edges. Further, we will call a path \emph{positive} if its weight is positive. We define negative and zero paths similarly.

Take an arena with up to $n$ nodes where all nodes of Min have out-degree 1 (that is, essentially Max is the one to move everywhere). First, we can remove all the nodes from where one can reach a positive cycle. Indeed, Max has a positional winning strategy from these nodes (Max can go to the closest simple positive cycle, and then start rotating over it forever). Here it is important that our arena is one-player. Two-player arenas might have positive cycles, but Max might be unable to stay on them.

Now the only way Max can win is by making the sum of the weights equal to $0$ infinitely many times. As a first attempt, consider an ``illegal'' memory skeleton $\mathcal{M}_\infty$, which simply stores the sum of the weights along the current play. It is illegal since the sum of the weights can be arbitrarily large, so $\mathcal{M}_\infty$ will have infinitely many states. Still, our winning condition for Max can be reformulated in terms of $\mathcal{M}_\infty$. Indeed, Max just has to bring $\mathcal{M}_\infty$ into a state ``the current sum is $0$'' infinitely many times. Notice that this is a parity condition in the product of our initial arena and the memory skeleton $\mathcal{M}_\infty$. Since parity games are positionally determined~\cite{zielonka1998infinite}, we have a uniform positional equilibrium in the product arena, and this gives a uniform $\mathcal{M}_\infty$-strategy equilibrium in the initial arena.

To turn this idea into a proof, we ``truncate'' $\mathcal{M}_\infty$. For arenas with up to $n$ nodes we consider a memory skeleton $\mathcal{M}_n$, which stores the  current sum of the weights while its absolute value is at most $n$; if it exceeds $n$, our memory skeleton comes into a special invalid state. Observe that such memory skeleton requires just $2n + 2$ states.

 We now make use of the fact that w.l.o.g~our arena has no positive cycles. Since our weights are $\pm1$, there is no path of weight larger than $n$. Indeed, any path can be decomposed into cycles and a simple path. The contribution of cycles is non-negative, and the contribution of a simple path is at most $n$, just because its length is at most $n$. So the current sum of the weights can never become larger than $n$. It can become smaller than $-n$, and in this case Max looses (he can never make it equal to $0$ again). So the goal of Max is, first, to avoid a state ``the current sum exceeded $n$ in the absolute value'', and second,  to reach a state ``the current sum is $0$'' infinitely many times. This is a parity condition in the product of our initial arena and the memory skeleton $\mathcal{M}_n$. Therefore, we get a uniform $\mathcal{M}_n$-strategy equilibrium in our initial arena.

\subsection{Theorem \ref{thm:example}}

Let the set of colors be $C = \{0, 1\}$. Fix a set $T\subseteq \mathbb{Z}^{+}$. Define a payoff function $\vphi\colon\{0, 1\}^\omega\to\{0, 1\}$ by setting $\vphi(\alpha) = 1$ for $\alpha = \alpha_1 \alpha_2 \alpha_3\ldots \in\{0, 1\}^\omega$ if and only if at least one of the following two conditions holds:
\begin{itemize}
\item $\alpha$ contains only finitely many $0$'s;
\item for some $t\in T$, the sequence $\alpha$ contains the word $01^t0$.
\end{itemize}
We show that, under some condition on $T$, the payoff function $\vphi$ is not arena-independent finite-memory determined, but belongs to $\oneFMD(f)$ for some $f\colon\mathbb{Z}^+\to\mathbb{Z}^+, f\notin \Omega(n)$. This condition is called \emph{isolation}. Roughly speaking, it requires that there are infinitely many elements in $T$ such that far to the left and to the right of them there are no other elements of $T$. More precisely, $T\subseteq \mathbb{Z}^+$ is isolated if there are infinitely many $k\in T$ such that $l \notin T$ for all $k/2 < l < k^4$, $l\neq k$. We call such $k$ \emph{isolated elements} of $T$.

From now on, we fix any  isolated set $T$, for example, $T = \{2^{4^n} \mid n\in\mathbb{Z}^+\}$. To show that $\vphi\in\oneFMD(f)$ for some $f\colon\mathbb{Z}^+\to\mathbb{Z}^+, f\notin \Omega(n)$, we construct, for every $k$, the following memory skeleton $\mathcal{M}_k$. It simply counts the number of 1's after the last 0. If this number exceeds $k$, it stops counting (it just remembers a fact that there are more than $k$ ones after the last $0$). Now, when our memory skeleton receives a $0$, there are two cases. If the current value of the counter is some number from $T\cap [1, k]$, then $\mathcal{M}_k$ transits into a  special ``winning state'', and stays in it forever. Otherwise, it resets the counter to $0$ and starts counting again.

Note that $\mathcal{M}_k$ can be realized with $k + O(1)$ states.
We show that if $k$ is an isolated element of $T$, then all arenas (even two-player) with up to $k^2$ nodes have a uniform $\mathcal{M}_k$-strategy equilibrium. This will show that $\vphi\in \FMD(f)$ for some function $f$ such that $f(n) \le 2\sqrt{n}$ for infinitely many $n$.

Consider any arena with up to $k^2$ nodes. We define an auxiliary game in which Max wins if either $\mathcal{M}_k$ was brought to the  ``winning state'' or there were just finitely many 0's in the play. Note that if Max wins in the auxiliary game, then Max wins w.r.t.~$\vphi$. The auxiliary game, however, is not entirely equivalent to $\vphi$, because a play can be winning for Min in the auxiliary game but loosing for Min w.r.t.~$\vphi$ (if this play contains $01^t0$ for some $t\in T, t > k$). Still, it holds that if Min can win in the auxiliary game, then Min can also win w.r.t.~$\vphi$. To prove this claim, we notice that the auxiliary game is a parity game in the product of our initial arena and the memory skeleton $\mathcal{M}_k$. So if Min can win in it, then Min can do so via some positional strategy $\tau$ in the product arena. We observe that $\tau$ is also winning w.r.t.~$\vphi$. Indeed, otherwise there is a play against $\tau$ which contains $01^t0$ for some $t\in T, t > k$. Since $k$ is an isolated element of $T$, we have $t \ge k^4$. Therefore, as the size of the product arena is $(k + O(1)) \cdot k^2 < k^4$, there must be a cycle which is consistent with $\tau$ and which consists entirely of 1's. But then Max can win against $\tau$ in the auxiliary game, contradiction.

As we pointed out, the auxiliary game is a parity game in the product of our arena with $\mathcal{M}_k$. Thus, it has a positional equilibrium there. This positional equilibrium translates into an $\mathcal{M}_k$-strategy equilibrium in the initial arena. Finally, as shown in the previous paragraph, any equilibrium in the auxiliary game is also an equilibrium w.r.t.~$\vphi$.

Showing that $\vphi$ is not arena-independent finite-memory determined is much easier. Take an isolated element $k\in T$. The idea is to construct an arena with a node which ``cuts'' the word $01^k0$ in $\Omega(k)$ different ways near the middle. Due to isolation, the only way for Max to win in this arena is to go through one of the cuts. However, Min can choose any of the cuts, so Max needs $\Omega(k)$ states to distinguish between different cuts. Since $k$ can be arbitrarily large, this shows that no single memory skeleton can be sufficient for $\vphi$ in all arenas.

\section{Warm-Up: Proof of Proposition \ref{positional_lifting}}
\label{sec:special_case}
The proof is by induction on the number of edges of an arena. More precisely, we are proving by induction on $m$ the following claim: for every $m$  every arena with $m$ edges and at most $N$ nodes has a uniform positional equilibrium.

The induction base ($m = 1$) is trivial (any arena with one edge is one-player and has exactly one node, so we can just refer to the assumption of the lemma). We proceed to the induction step. Take an arena 
$\mathcal{A} = \langle V, V_\Max, V_\Min, E, \source,\target,\col\rangle$ with at most $N$ nodes and assume that all arenas with at most $N$ nodes and with fewer edges than $\mathcal{A}$ have a uniform positional equilibrium. 
 We prove the same for $\mathcal{A}$. Since the set of uniform equilibria is a Cartesian product by Lemma \ref{cartesian}, 
it is enough to establish the following two claims:
\begin{itemize}
\item \textbf{\emph{(a)}} in $\mathcal{A}$ there exists a uniform equilibrium including a positional strategy of Max;
\item \textbf{\emph{(b)}} in $\mathcal{A}$ there exists a uniform equilibrium including a positional strategy of Min.
\end{itemize}
We only show \textbf{\emph{(a)}}, a proof of \textbf{\emph{(b)}} is similar.

We may assume that $\mathcal{A}$ is not one-player (otherwise we are done due to the assumptions of the lemma). Hence there exists a node $w\in V_\Max$ with out-degree at least $2$. Partition the set $E(w) = \{e\in E\mid \source(e) = w\}$ into two non-empty disjoint subsets $E_1(w)$ and $E_2(w)$. 
 Define two new arenas $\mathcal{A}_1$ and $\mathcal{A}_2$. The arena $\mathcal{A}_1$ is obtained from $\mathcal{A}$ by deleting edges from the set $E_2(w)$. Similarly, the arena $\mathcal{A}_2$ is obtained from $\mathcal{A}$ by deleting edges from the set $E_1(w)$. So in $\mathcal{A}_i$ for $i = 1, 2$ the set of edges with the source in $w$ is $E_i(w)$.

Both $\mathcal{A}_1$ and $\mathcal{A}_2$ have fewer edges than $\mathcal{A}$. So both these arenas have a uniform positional equilibrium. Let $(\sigma_i, \tau_i)$ be a uniform positional equilibrium in $\mathcal{A}_i$ for $i = 1, 2$. We will first define two auxiliary strategies $\tau_{12}$ and $\tau_{21}$ of Min; then we will show that either $(\sigma_1, \tau_{12})$ or $(\sigma_2, \tau_{21})$ is a uniform equilibrium in $\mathcal{A}$. After that \textbf{\emph{(a)}} will be proved.

\medskip

 Strategies $\tau_{12}$ and $\tau_{21}$ will not be positional. In a sense, they are combinations of $\tau_1$ and $\tau_2$. In both strategies  Min has a counter $I$ which can only take two values, $1$ and $2$. The counter $I$ indicates to Min which of the strategies $\tau_1$ or $\tau_2$ to use. I.e., whenever Min should make a move from a node $v\in V_\Min$, he uses an edge $\tau_I(v)$.
  The value of $I$ changes each time in the node $w$ Max uses an edge not from a set $E_{I}(w)$. It only remains to specify the initial value of $I$. There are two ways to do this, one will give us strategy $\tau_{12}$, and the other will give $\tau_{21}$. More specifically, in $\tau_{12}$ the initial value of $I$ is $1$ and in $\tau_{21}$ the initial value of $I$ is $2$.

It is not hard to see that $\tau_{12}$ is a uniformly optimal response to $\sigma_1$ and $\tau_{21}$ is a uniformly optimal response to $\sigma_2$. For instance, let us show this for $\tau_{12}$ and $\sigma_1$. By definition, $\tau_1$ is a uniformly optimal response to $\sigma_1$ in the arena $\mathcal{A}_1$, and hence also in the arena $\mathcal{A}$ (because any play against $\sigma$ takes place inside $\mathcal{A}_1$). It remains to notice that $\tau_{12}$ plays exactly as $\tau_1$ against $\sigma_1$. Indeed, $\sigma_1$ never uses edges from $E_2(w)$, so the counter $I$ always equals $1$ against $\sigma_1$.

It remains to show that either $\sigma_1$ is a uniformly optimal response to $\tau_{12}$ or $\sigma_2$ is a uniformly optimal response to $\tau_{21}$ (in the arena $\mathcal{A}$). We derive it from the assumption of the lemma applied to an auxiliary one-player arena $\mathcal{B}$ with at most $2N - 1$ nodes.

Namely, we define $\mathcal{B}$ as follows. Recall that in our notation $(\mathcal{A}_1)_{\tau_1}$ and $(\mathcal{A}_2)_{\tau_2}$ stand for two arenas obtained from, respectively, $\mathcal{A}_1$ and $\mathcal{A}_2$ by throwing away edges that are inconsistent with, respectively, $\tau_1$ and $\tau_2$. Consider an arena consisting of two ``independent'' parts one of which coincides with $(\mathcal{A}_1)_{\tau_1}$ and the other with  $(\mathcal{A}_2)_{\tau_2}$ (``independent'' means that there are no edges between the parts).  From each part take a node corresponding to the node $w$. Then merge these two nodes into a single one. The resulting arena with $2|V| - 1 \le 2N - 1$ nodes will be $\mathcal{B}$.

\begin{figure}[h!]
\centering
  \includegraphics[width=0.6\textwidth]{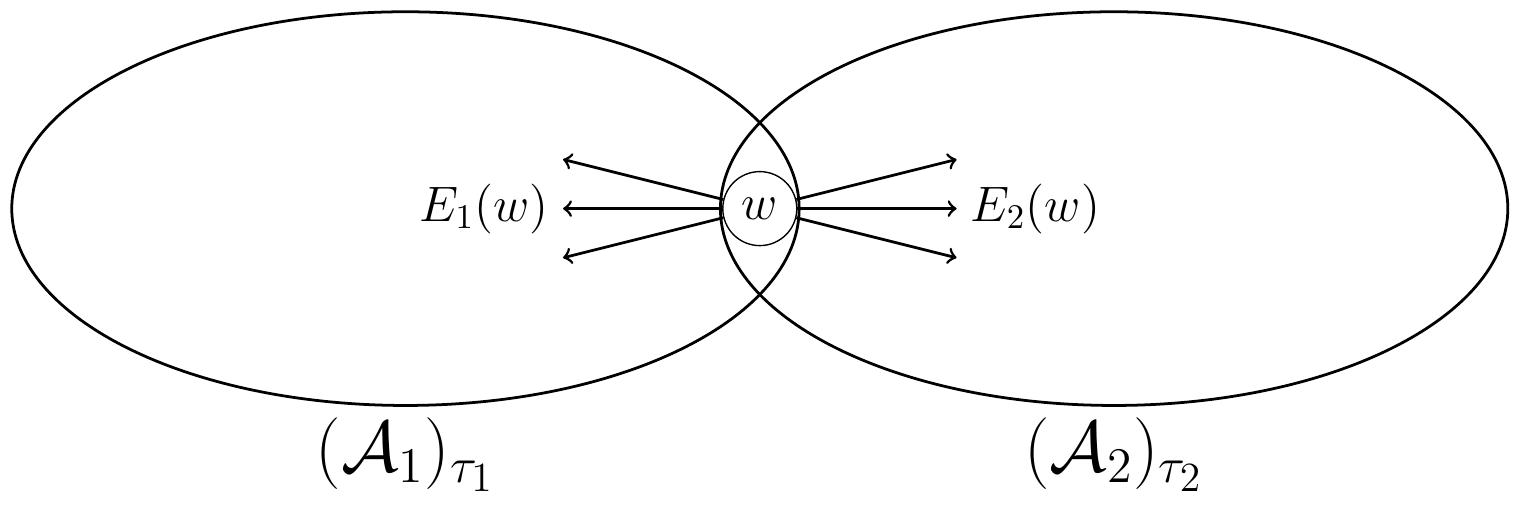}
  \caption{Arena $\mathcal{B}$.}
\label{two_clouds}
\end{figure}

For each node of $\mathcal{A}$ there are two ``copies'' of it in $\mathcal{B}$ -- one from $(\mathcal{A}_1)_{\tau_1}$ and the other from $(\mathcal{A}_2)_{\tau_2}$. We will call copies of the first kind \emph{left copies} and copies of the second kind \emph{right copies}. Note that the left and the right copy of $w$ is the same node in $\mathcal{B}$. Any other node of $\mathcal{A}$ has two distinct copies. Now, by the \emph{prototype} of a node $v^\prime$ of $\mathcal{B}$ we mean a node $v$ of $\mathcal{A}$ of which $v^\prime$ is a copy.

Note that in $\mathcal{B}$ all nodes of Min have out-degree $1$ (because they do so inside $(\mathcal{A}_1)_{\tau_1}$ and $(\mathcal{A}_2)_{\tau_2}$, and the only node of $\mathcal{B}$ which was obtained by merging two nodes is a node of Max). Thus, $\mathcal{B}$ is a one-player arena.

An important feature of $\mathcal{B}$ is that it can ``emulate'' any play against $\tau_{12}$ and $\tau_{21}$ in $\mathcal{A}$. Formally,
\begin{lem}
\label{translation}
For any infinite path $h$ in $\mathcal{A}$ which is consistent with $\tau_{12}$ there exists an infinite path $h^\prime$ in $\mathcal{B}$ with $\col(h^\prime) = \col(h)$ and with the source in the left copy of $\source(h)$. Similarly, for any infinite path $h$ in $\mathcal{A}$ which is consistent with $\tau_{21}$ there exists an infinite path $h^\prime$ in $\mathcal{B}$ with $\col(h^\prime) = \col(h)$ and with the source in the right copy of $\source(h)$.
\end{lem}
\begin{proof}
We only give an argument for $\tau_{12}$, the argument for $\tau_{21}$ is similar. We construct $h^\prime$ from the left copy of $\source(h)$ by always moving in the same ``local direction'' as $h$. There will be no problem with that for the nodes of Max because they have the same set of out-going edges in $\mathcal{B}$ as their prototypes have in $\mathcal{A}$. Now, for the nodes of Min we should be more accurate. The path $h$ is consistent with $\tau_{12}$, so from the nodes of Min it applies either $\tau_1$ or $\tau_2$. Now, in $\mathcal{B}$ strategy $\tau_1$ is available only in the left ellipse of Figure \ref{two_clouds}, and $\tau_2$ is available only in the right ellipse. So each time $h$ wants to apply $\tau_1$, the path $h^\prime$ should be in the left ellipse. Similarly, each time $h$ wants to apply $\tau_2$, the path $h^\prime$ should be in the right ellipse. Initially, until its counter changes, $\tau_{12}$ applies $\tau_1$, and correspondingly $h^\prime$ starts in the left ellipse.
 Now, each time $\tau_{12}$ switches to $\tau_2$, it does so because Max used an edge from $E_2(w)$ in $w$. Correspondingly, $h^\prime$ enters the right ellipse at this moment. Similarly, whenever $\tau_{12}$ switches back to $\tau_1$, the path $h^\prime$ returns to the left ellipse.
\end{proof}

Note that in $\mathcal{B}$ Min has exactly one strategy. We denote it by $T$.
The arena $\mathcal{B}$ is one-player and has at most $2N -1$ nodes, so by the assumption of the lemma there is a uniform positional equilibrium $(\widehat{\Sigma}, T)$ in it. We claim the following:
\begin{itemize}
\item if $\widehat{\Sigma}$ applies an edge from $E_1(w)$ in $w$, then $\sigma_1$ is a uniformly optimal response to $\tau_{12}$ in $\mathcal{A}$;
\item if $\widehat{\Sigma}$ applies an edge from $E_2(w)$ in $w$, then $\sigma_2$ is a uniformly optimal response to $\tau_{21}$ in $\mathcal{A}$.
\end{itemize}
We only show the first claim, the proof of the second one is analogous. Consider a restriction of $\widehat{\Sigma}$ to the left ellipse of $\mathcal{B}$. This defines a positional  strategy $\widehat{\sigma}$ of Max in $\mathcal{A}$. Note that in each node of $\mathcal{A}$ the strategy $\sigma_1$ is at least as good against $\tau_{12}$ as $\widehat{\sigma}$. Indeed,  $\sigma_1(w), \widehat{\sigma}(w)\in E_1(w)$. Hence $\sigma_1, \widehat{\sigma}$ are strategies in the arena  $\mathcal{A}_1$, where $\sigma_1$ is a uniformly optimal response to $\tau_1$. It remains to notice that $\tau_{12}$ plays exactly as $\tau_1$ against  $\sigma_1$ and $\widehat{\sigma}$ since these two strategies of Max never use edges from $E_2(w)$.

Therefore, it is enough to show that $\widehat{\sigma}$ is a uniformly optimal response to $\tau_{12}$ in $\mathcal{A}$. Take any node $v\in V$ and any play $h$ against $\tau_{12}$ from $v$. Our goal is to show that the play of $\widehat{\sigma}$ and $\tau_{12}$ from $v$ is at least as good from the Max's perspective as $h$. Now, by Lemma \ref{translation} some infinite path $h^\prime$ from the left copy of $v$ is colored exactly as $h$. On the other hand, the play of $\widehat{\Sigma}$ and $T$ from the left copy of $v$ is at least as good for Max as $h^\prime$ (and hence as $h$). This is because $h^\prime$ is consistent with $T$ (as there are simply no other strategies of Min in $\mathcal{B}$) and because $(\widehat{\Sigma}, T)$ is an equilibrium. It remains to note that the play of $\widehat{\Sigma}$ and $T$ from the left copy of $v$ is colored exactly as the play of $\widehat{\sigma}$ and $\tau_{12}$ from $v$. Indeed, as we have already observed, $\tau_{12}$ plays exactly as $\tau_1$ against $\widehat{\sigma}$. On the other hand, the play of $\widehat{\Sigma}$ and $T$ can never leave the left ellipse as $\widehat{\Sigma}$ points to the left in $w$. Moreover, restrictions of these strategies to the left ellipse coincide with $\widehat{\sigma}$ and $\tau_1$; for $\widehat{\Sigma}$ this is just by definition and for $T$ this is because the left ellipse coincides with the arena $(\mathcal{A}_1)_{\tau_1}$

\section{Proof of Theorem \ref{lifting}}
\label{sec:main_proof}

\subsection{$T$-wise equilibria}
\label{subsec:cartesian}
Before diving into details of our proof of Theorem \ref{lifting}, we have to generalize the notion of a uniform equilibrium. Take any arena \[\mathcal{A} = \langle V, V_\Max, V_\Min, E, \source,\target,\col\rangle\] and any payoff function $\vphi\colon C^\omega\to\mathcal{W}$.
Fix a subset $T\subseteq V$, a strategy $\sigma$ of Max and a strategy $\tau$ of Min. We say that $\sigma$ is a \textbf{$T$-wise optimal response} to $\tau$ w.r.t.~$\vphi$ if for all $v\in T$ and for all infinite $h\in \Cons(v,\tau)$ we have $\vphi\circ \col\big(h(v, \sigma,\tau)\big) \ge \vphi\circ\col(h)$. 
Similarly, we call $\tau$ a \textbf{$T$-wise optimal response} to $\sigma$ w.r.t.~$\vphi$ if for all $v\in T$ and for all infinite $h\in \Cons(v,\sigma)$ we have $\vphi\circ \col\big(h(v, \sigma,\tau)\big) \le \vphi\circ\col(h)$.
Finally, we call a pair $(\sigma, \tau)$ a \textbf{$T$-wise equilibrium} w.r.t.~$\vphi$ if $\sigma$ and $\tau$ are  $T$-wise optimal responses to each other.

When $T = V$ is the whole set of nodes, then $T$-wise equilibria are uniform equilibria, and vice versa. Thus, the following lemma generalizes Lemma \ref{cartesian}.

\begin{lem}
\label{cartesian2}
 For any arena $\mathcal{A} = \langle V, V_\Max, V_\Min, E, \source,\target,\col\rangle$, for any payoff function $\vphi$, and for any subset $T\subseteq V$, the set of $T$-wise equilibria in $\mathcal{A}$ w.r.t.~$\vphi$ is a Cartesian product.
\end{lem}
\begin{proof}
It is sufficient to show the following: if $(\sigma_1, \tau_1)$ and $(\sigma_2, \tau_2)$ are $T$-wise equilibria, then so is $(\sigma_1, \tau_2)$. That is, our goal is to show that $\sigma_1$ is a $T$-wise optimal response to $\tau_2$, and that $\tau_2$ is a $T$-wise optimal response to $\sigma_1$. We only prove the first claim, the second one can be proved similarly. Take any $v\in T$ and any infinite $h\in\Cons(v, \tau_2)$. We have to show that $\vphi\circ \col\big(h(v, \sigma_1,\tau_2)\big) \ge \vphi\circ\col(h)$. We first show that $\vphi\circ \col\big(h(v, \sigma_1,\tau_2)\big) = \vphi\circ \col\big(h(v, \sigma_1,\tau_1)\big) = \vphi\circ \col\big(h(v, \sigma_2,\tau_2)\big)$. Indeed,
\begin{align*}
\vphi\circ \col\big(h(v, \sigma_1,\tau_1)\big) &\ge \vphi\circ \col\big(h(v, \sigma_2,\tau_1)\big) \ge \vphi\circ \col\big(h(v, \sigma_2,\tau_2)\big) \\
&\ge \vphi\circ \col\big(h(v, \sigma_1,\tau_2)\big) \ge \vphi\circ \col\big(h(v, \sigma_1,\tau_1)\big).
\end{align*}
The first inequality here holds because $\sigma_1$ is a $T$-wise optimal response to $\tau_1$. The second inequality here holds because $\tau_2$ is a $T$-wise optimal response to $\sigma_2$. The third inequality here holds because $\sigma_2$ is a $T$-wise optimal response to $\tau_2$. The fourth inequality here holds because $\tau_1$ is a $T$-wise optimal response to $\sigma_1$.

As we have shown, $\vphi\circ \col\big(h(v, \sigma_1,\tau_2)\big) = \vphi\circ \col\big(h(v, \sigma_2,\tau_2)\big)$.  In turn, since $h\in \Cons(v, \tau_2)$, and since $\sigma_2$ is a $T$-wise optimal response to $\tau_2$, we have that $\vphi\circ \col\big(h(v, \sigma_2,\tau_2)\big) \ge \vphi\circ\col(h)$. Therefore, we get $\vphi\circ \col\big(h(v, \sigma_1,\tau_2)\big) \ge \vphi\circ\col(h)$.

\end{proof}

\subsection{Plan of the proof}
 We reduce Theorem \ref{lifting} to a statement about positional strategies (namely, to Lemma \ref{main_lemma} below). First we need a classical concept of \emph{product arenas}.

\begin{defi}[Product arenas]
\label{product_arena}
Let $\mathcal{M} = \langle M, m_{init}, \delta\colon M\times C\to M\rangle$ be a memory skeleton and $\mathcal{A} = \langle V, V_\Max, V_\Min, E, \source, \target, \col\rangle$ be an arena. Then $\mathcal{M}\times\mathcal{A}$ stands for an arena, where
\begin{itemize}
\item the set of nodes is $M\times V$;
\item the set of Max's nodes is $M\times V_\Max$;
\item the set of Min's nodes is $M\times V_\Min$;
\item the set of edges is $M\times E$;
\item the source function is defined as follows: $\source((m, e)) = (m, \source(e))$;
\item the target function is defined as follows: $\target((m, e)) = \big(\delta(m, \col(e)), \target(e) \big)$;
\item the coloring function is defined as follows: $\col((m, e)) = \col(e)$.
\end{itemize}
\end{defi}

The following is a standard observation that product arenas reduce finite-memory determinacy to positional determinacy.

\begin{obs}
\label{prod_obs}
Let \[\mathcal{M} = \langle M, m_{init}, \delta\colon M\times C\to M\rangle\] be a memory skeleton and \[\mathcal{A} = \langle V, V_\Max, V_\Min, E, \source, \target, \col\rangle\] be an arena. Then for every $S\subseteq V$ the following holds: if $\mathcal{M}\times\mathcal{A}$ has an $(\{m_{init}\}\times S)$-wise positional equilibrium, then $\mathcal{A}$ has an $S$-wise  $\mathcal{M}$-strategy equilibrium.
\end{obs}

Its proof can be found in Appendix \ref{app:prod_obs}.

 Next we introduce one more concept which we need for the reduction, namely, one of \emph{$\mathcal{M}$-triviality}.

\begin{defi}
\label{consistent}
 Let $\mathcal{M}  = \langle M, m_{init}, \delta\colon M\times C \to M\rangle$ be a memory skeleton. A pair $(\mathcal{A}, f)$ of an arena $\mathcal{A} = \langle V, V_\Max, V_\Min, E, \source, \target, \col\rangle$ and a function $f\colon V\to M$ is called \textbf{$\mathcal{M}$-trivial} if for every $e\in E$ it holds that
$\delta\big(f(\source(e)), \col(e)\big) = f(\target(e))$.
\end{defi}

Informally, $f$ is a mapping from $\mathcal{A}$ to the transition graph of $\mathcal{M}$ which takes into account the colors of the edges. Of course, there are arenas that belong to no $\mathcal{M}$-trivial pair.
We observe that  $\mathcal{M}$-strategies, in a sense, degenerate to positional ones in $\mathcal{M}$-trivial pairs.

\begin{obs}
\label{cons_obs}
Let $\mathcal{M}  = \langle M, m_{init}, \delta\colon M\times C \to M\rangle$ be a memory skeleton. Then for every $\mathcal{M}$-trivial pair $(\mathcal{A}, f)$ the following holds: if $\mathcal{A}$ has a uniform $\mathcal{M}$-strategy equilibrium, then $\mathcal{A}$ has an $f^{-1}(m_{init})$-wise positional equilibrium.
\end{obs}
\begin{proof}
Note that  for any finite path $h$ in $\mathcal{A}$ we have:
\[\delta(f(\source(h)), \col(h)) = f(\target(h)).\]
Indeed, this holds by definition as long as $h$ is a single edge; for longer $h$ this can be easily proved by induction on $|h|$.

To show the observation, we simply show that any $\mathcal{M}$-strategy coincides with some positional  one on all plays that start in the nodes of $f^{-1}(m_{init})$. 
 Indeed, a move of an $\mathcal{M}$-strategy in a position $h$ depends solely on $\target(h)$ and $\delta(m_{init}, \col(h))$. However, $\delta(m_{init}, \col(h)) = \delta\big(f(\source(h)), \col(h)\big) = f(\target(h))$ for all $h$ with $\source(h) \in f^{-1}(m_{init})$. In other words, for all such $h$ a move of an $\mathcal{M}$-strategy in $h$ is a function only of $\target(h)$, as required.
\end{proof}

We are ready to formulate a statement about positional strategies to which we reduce Theorem \ref{lifting}.
\begin{lem}
\label{main_lemma}
Let $\mathcal{M}  = \langle M, m_{init}, \delta\colon M\times C \to M\rangle$ be a memory skeleton. Assume that for every $\mathcal{M}$-trivial pair $(\mathcal{A}, f)$ such that $\mathcal{A}$ is one-player and has at most $2N - 1$ nodes there exists an $f^{-1}(m_{init})$-wise positional equilibrium in $\mathcal{A}$.

Then for every $\mathcal{M}$-trivial pair $(\mathcal{A}, f)$ such that $\mathcal{A}$ has at most $N$ nodes there exists an $f^{-1}(m_{init})$-wise positional equilibrium in $\mathcal{A}$.
\end{lem}

\subsection{Derivation of Theorem \ref{lifting} from Lemma \ref{main_lemma}}

Let \[\mathcal{A} = \langle V, V_\Max, V_\Min, E, \source, \target, \col\rangle\] be an arena with at most $n$ nodes. Our goal is to show that $\mathcal{A}$ has a uniform $\mathcal{M}$-strategy equilibrium. By Observation \ref{prod_obs}, it is sufficient to show that the arena $\mathcal{M}\times \mathcal{A}$ has an $\{m_{init}\}\times V$-wise positional equilibrium. It is easy to see that a pair $(\mathcal{M}\times\mathcal{A}, f)$, where 
\[f\colon M\times V\to M, \qquad f((m, v)) = m,\]
is an $\mathcal{M}$-trivial pair, by definition of $\mathcal{M}\times \mathcal{A}$. Observe that $\{m_{init}\}\times V = f^{-1}(m_{init})$, so we only have to show that $\mathcal{M}\times \mathcal{A}$ has an $f^{-1}(m_{init})$-wise positional equilibrium. Since $\mathcal{M}\times \mathcal{A}$ has at most $|\mathcal{M}| \cdot n$ nodes, it remains to explain why the assumption of Lemma \ref{main_lemma} holds for $N = |\mathcal{M}| \cdot n$.

By the assumption of Theorem \ref{lifting} all one-player arenas with at most $2|\mathcal{M}|\cdot n - 1 = 2N - 1$ nodes have a uniform $\mathcal{M}$-strategy equilibrium. In particular, this applies to any one-player arena $\mathcal{A}^\prime$ with at most $2N - 1$ nodes belonging to some $\mathcal{M}$-trivial pair $(\mathcal{A}^\prime, f)$. By Observation \ref{cons_obs} this means that all such $\mathcal{A}^\prime$ have a $f^{-1}(m_{init})$-wise positional equilibrium, as required.

\subsection{Proof of Lemma \ref{main_lemma}}
We use the same technique and terminology as in Section \ref{sec:special_case}.
 We are now proving by induction on $m$ the following claim: for every $m$ and for every $\mathcal{M}$-trivial pair $(\mathcal{A}, f)$ such that $\mathcal{A}$ has $m$ edges and at most $N$ nodes there exists an  $f^{-1}(m_{init})$-wise positional equilibrium in $\mathcal{A}$.

Induction base ($m = 1$) again requires no argument, and we proceed to the induction step. Consider any $\mathcal{M}$-trivial pair 
$(\mathcal{A}, f)$, where $\mathcal{A} = \langle V, V_\Max, V_\Min, E, \source, \target, \col\rangle$ has at most $N$ nodes. Our goal is to show that $\mathcal{A}$ has an $f^{-1}(m_{init})$-wise positional equilibrium, provided that an analogous claim is already proved  for all $\mathcal{M}$-trivial pairs $(\mathcal{A}^\prime, f^\prime)$ in which $\mathcal{A}^\prime$ has at most $N$ nodes and fewer edges than $\mathcal{A}$.
Since the set of $f^{-1}(m_{init})$-wise equilibria is a Cartesian product by Lemma \ref{cartesian2},
it is enough to establish the following two claims:
\begin{itemize}
\item \textbf{\emph{(a)}} in $\mathcal{A}$ there exists an $f^{-1}(m_{init})$-wise equilibrium including a positional strategy of Max;
\item \textbf{\emph{(b)}} in $\mathcal{A}$ there exists an $f^{-1}(m_{init})$-wise equilibrium including a positional strategy of Min.
\end{itemize}
We only show \textbf{\emph{(a)}}, a proof of \textbf{\emph{(b)}} is similar. As before, we may assume that $\mathcal{A}$ is not one-player so that there exists a node $w\in V_\Max$ with out-degree at least $2$. We partition the set of its out-going edges into two disjoint non-empty sets $E_1(w)$ and $E_2(w)$. Then we define arenas $\mathcal{A}_1$ and $\mathcal{A}_2$ exactly as in Section \ref{sec:special_case}. Since $(\mathcal{A}, f)$ is an $\mathcal{M}$-trivial pair, then so are pairs $(\mathcal{A}_1, f)$ and $(\mathcal{A}_2, f)$.
 Indeed, $\mathcal{A}_1$ and $\mathcal{A}_2$ were obtained by simply throwing away some edges of $\mathcal{A}$. The remaining edges satisfy the definition of $\mathcal{M}$-triviality with respect to $f$ just because they do so inside $\mathcal{A}$. 

Note that $\mathcal{A}_1$ and $\mathcal{A}_2$ both have fewer edges than $\mathcal{A}$ and at most as many nodes. So by the induction hypothesis both these arenas have an $f^{-1}(m_{init})$-wise positional equilibrium. Let $(\sigma_1, \tau_1)$ be an $f^{-1}(m_{init})$-wise positional equilibrium in $\mathcal{A}_1$ and $(\sigma_2, \tau_2)$ be an $f^{-1}(m_{init})$-wise positional equilibrium in $\mathcal{A}_2$. Next, we define two auxiliary strategies $\tau_{12}$ and $\tau_{21}$ of Min exactly as in Section \ref{sec:special_case}. Our goal is to show that either $(\sigma_1, \tau_{12})$ is an $f^{-1}(m_{init})$-wise equilibrium in $\mathcal{A}$ or $(\sigma_2,\tau_{21})$ is an $f^{-1}(m_{init})$-wise equilibrium in $\mathcal{A}$.  

By the same argument as in Section \ref{sec:special_case}, we have that $\tau_{12}$ is an $f^{-1}(m_{init})$-wise optimal response to $\sigma_1$ and $\tau_{21}$ is an $f^{-1}(m_{init})$-wise optimal response to $\sigma_2$. The main challenge is to show the opposite for at least one of the pairs $(\sigma_1, \tau_{12})$ and $(\sigma_2, \tau_{21})$.

 For that we define a one-player arena $\mathcal{B}$ exactly as in Section \ref{sec:special_case} (see Figure \ref{two_clouds}). It has $2|V| - 1\le 2N - 1$ nodes.  
We will apply the assumption of Lemma \ref{main_lemma} to $\mathcal{B}$. More precisely, this will be done for some $\mathcal{M}$-trivial pair which includes $\mathcal{B}$. For that we define the following mapping $g$ from the set of nodes of $\mathcal{B}$ to the set of states of $\mathcal{M}$. Namely, if $v^\prime$ is a node of $\mathcal{B}$, we set $g(v^\prime) = f(v)$, where $v$ is the prototype of $v^\prime$. Observe that $(\mathcal{B}, g)$ is an $\mathcal{M}$-trivial pair. Indeed any edge of $\mathcal{B}$ is between two nodes whose prototypes are connected in $\mathcal{A}$ by an edge of the same color. Thus, by the assumption of Lemma \ref{main_lemma}, the arena $\mathcal{B}$ has a $g^{-1}(m_{init})$-wise positional equilibrium $(\widehat{\Sigma}, T)$ (as before, in $\mathcal{B}$ there are no strategies of Min other than $T$). It is sufficient to establish the following two claims:
\begin{itemize}
\item if $\widehat{\Sigma}$ applies an edge from $E_1(w)$ in $w$, then $\sigma_1$ is an $f^{-1}(m_{init})$-wise optimal response to $\tau_{12}$ in $\mathcal{A}$;
\item if $\widehat{\Sigma}$ applies an edge from $E_2(w)$ in $w$,  then $\sigma_2$ is an $f^{-1}(m_{init})$-wise optimal response to $\tau_{21}$ in $\mathcal{A}$.
\end{itemize}
A key observation here is that $g^{-1}(m_{init})$ is the union of the left copies of the nodes from $f^{-1}(m_{init})$ and the right copies of the nodes of $f^{-1}(m_{init})$. In fact, for a proof of the first claim we only need a fact that $g^{-1}(m_{init})$ includes  all the left copies of the nodes from $f^{-1}(m_{init})$. Correspondingly, only the right copies of  $f^{-1}(m_{init})$  are relevant for a proof of the second claim.

We only show the first claim, the second one can be proved similarly. As in Section \ref{sec:special_case}, the argument is carried out through a positional strategy $\widehat{\sigma}$ of Max in $\mathcal{A}$ obtained by restricting $\widehat{\Sigma}$ to the left ellipse. First we observe that in any node from $f^{-1}(m_{init})$ the strategy $\sigma_1$ is at least as good against $\tau_{12}$ as $\widehat{\sigma}$. Indeed, both $\sigma_1$ and $\widehat{\sigma}$ are strategies in $\mathcal{A}_1$ whereas $\sigma_1$ is an $f^{-1}(m_{init})$-wise optimal response to $\tau_1$ in $\mathcal{A}_1$ by definition. On the other hand, $\tau_{12}$ plays against $\sigma_1$ and $\widehat{\sigma}$ exactly as $\tau_1$.

It remains to show that $\widehat{\sigma}$ is an optimal response to $\tau_{12}$ in any node from $f^{-1}(m_{init})$. This can be done by exactly the same argument as in the last paragraph of  Section \ref{sec:special_case}. A difference is that now we have a weaker assumption about $\widehat{\Sigma}$; namely, we only know that $\widehat{\Sigma}$ is optimal in the nodes from $g^{-1}(m_{init})$ (while before it was optimal everywhere in $\mathcal{B}$). Correspondingly, we are proving a weaker statement. Namely, instead of proving that $\widehat{\sigma}$ is an optimal response to $\tau_{12}$ everywhere in $\mathcal{A}$, we are only proving this for all $v\in f^{-1}(m_{init})$. It can be checked that in the argument for a specific $v$ we only require optimality of $\widehat{\Sigma}$ in the left copy of $v$; so if $v\in f^{-1}(m_{init})$, then its left copy is in $g^{-1}(m_{init})$ so that the argument still works.

\section{Proof of Theorem \ref{thm:sharp}}
\label{sec:sharp}

Let the set of colors be $C = \{-1, 1\}$. Define the payoff function $\psi \colon C^\omega \to \{0, 1\}$ as follows:
\[\psi(c_1 c_2 c_3 \ldots) = 1 \iff \mbox{either } \big(\lim\limits_{n\to\infty} \sum\limits_{i = 1}^n c_i = +\infty\big) \mbox{ or } \big(\sum\limits_{i = 1}^n c_i = 0 \mbox{ for infinitely many } n\big).\]
We assume the standard ordering on $\{0, 1\} = \psi(C^\omega)$ so that $1$ is interpreted as victory of Max and $0$ is interpreted as victory of Min. We show that $\psi\in\oneFMD(2n + 2)\setminus \FMD$. In fact, this payoff was already considered in~\cite{bouyer2020games}. It is shown there that $\psi\in\oneFMD\setminus \FMD$. Thus, our contribution is the inclusion $\psi\in\oneFMD(2n + 2)$.

 It will be more convenient to refer to the elements of $C$ as \emph{weights} rather than as colors. Correspondingly, by the weight of a path we will mean the sum of the weights of its edges. Further, we will call a path \emph{positive} if its weight is positive. We define negative and zero paths similarly (in fact, we will apply this terminology only to cycles).

Given $n\in\mathbb{Z}^+$, define a memory skeleton $\mathcal{M}_n$ which stores the current sum of the weights until its absolute value exceeds $n$ (in this case it comes into a special ``invalid'' state and stays in it forever). We will denote its normal states by integers from $-n$ to $n$, and we will denote its invalid state by $\bot$.
The number of states of $\mathcal{M}_n$ is $2n + 2$. The rest of the proof is organized as follows:
\begin{itemize}
\item in Subsection \ref{subsec:one}  we show that all one-player arenas with at most $n$ nodes have a uniform $\mathcal{M}_n$-strategy equilibrium w.r.t.~$\psi$ (and, thus, $\psi\in\oneFMD(2n + 2)$); 
\item  in Subsection \ref{subsec:two} we show that $\psi$ is not in finite-memory determined. Although it is already established  in~\cite[Section 3.4]{bouyer2020games}, we provide this argument for completeness.
\end{itemize}

\subsection{One-player arenas}
\label{subsec:one}

Consider a one-player arena $\mathcal{A} = \langle V, V_\Max, V_\Min, E, \source, \target, \col\rangle$ with at most $n$ nodes. By definition, either all nodes of Min have out-degree $1$ or all nodes of Max have out-degree $1$. We will consider these two cases separately. In both cases we will deal with a product arena $\mathcal{M}_n\times \mathcal{A}$ (see Definition \ref{product_arena}). Recall that the nodes of $\mathcal{M}_n\times\mathcal{A}$ are pairs of the form $(\mbox{state of } \mathcal{M}_n, \mbox{node of } \mathcal{A})$ so that it will be convenient to use the following notation for these nodes:
\[(-n, v), \ldots (-1, v), (0, v), (1, v), \ldots, (n, v), (\bot, v), \qquad v\in V.\]

\medskip

\emph{Case 1: all nodes of Min have out-degree $1$}. Let $V_+$ be the set of nodes of $\mathcal{A}$ from where one can reach a positive cycle. By a standard reasoning, inside $V_+$ Max has a positional strategy which guarantees that the sum of the weights goes to $+\infty$. This strategy is winning for Max with respect to $\psi$ as well.

Obviously, there are no edges from $V\setminus V^+$ to $V^+$. Thus, for the rest of the argument we may only deal with a restriction of $\mathcal{A}$ to $V\setminus V^+$. In other words, we may assume WLOG that $V^+$ is empty so that all cycles of $\mathcal{A}$ are non-positive. In particular, this implies that the weight of any path is at most $n$ (as it can be decomposed into a simple path and a union of cycles).

Now the sum of the weights cannot go to $+\infty$ so that Max can only win by making this sum equal to $0$ infinitely many times. Observe also that Max looses as long as the sum of the weights becomes smaller than $-n$. Indeed, in this case it is impossible to make it non-negative again (for that we would need a path of weight bigger than $n$).

These considerations show that our winning condition for Max is now equivalent to the following one: Max wins if the sum of the weights equals $0$ infinitely often but never exceeds $n$ in the absolute value. Equivalently (in terms of the memory skeleton $\mathcal{M}_n$) Max wins if $\mathcal{M}_n$ comes into state $0$ infinitely often but never comes into the invalid state. We can further simplify our winning condition by recalling $\mathcal{M}_n$ stays in the invalid state forever once this state is reached. So we can just forget about the requirement of avoiding the invalid state; as long as  $\mathcal{M}_n$ comes into state $0$ infinitely often, we automatically have that it never comes into the invalid state.

In terms of the product arena $\mathcal{M}_n\times\mathcal{A}$ this just means that Max wins from $w\in V$ if and only there exists
 an infinite path from $(0, w)$ in $\mathcal{M}_n\times\mathcal{A}$ which visits some node of the form $(0, v), v\in V$ infinitely often. Hence, this winning condition is just a parity game~\cite[Chapter 2]{obdrzalek2006algorithmic} with 2 priorities. Indeed, label all the nodes of the from $(0, v)$ by priority $2$ and all the other nodes by priority $1$. Observe that Max wins if and only if the largest priority visited infinitely many times is $2$.

 By positional determinacy of parity games~\cite{zielonka1998infinite} some positional strategy $\sigma$ of Max in $\mathcal{M}_n\times\mathcal{A}$ is winning for him wherever he has a winning strategy. Similarly to the proof of Observation \ref{prod_obs}, strategy $\sigma$ defines an $\mathcal{M}_n$-strategy $\Sigma$ of Max in $\mathcal{A}$ which is winning wherever Max has a winning strategy. This strategy $\Sigma$ (together with a unique strategy of Min) forms a uniform $\mathcal{M}_n$-strategy equilibrium in $\mathcal{A}$.

\medskip

\emph{Case 2: all nodes of Max have out-degree $1$}. Similarly to Case 1 we may assume WLOG that all cycles of $\mathcal{A}$ are non-negative. Indeed, in all  nodes from where one can reach a negative cycle Min can win by making the sum of the weights going to $-\infty$; moreover, he has a single positional strategy which does this for all these nodes. There are no edges to these nodes from the remaining ones; so we can restrict our arena to the set of nodes from where no negative cycle is reachable.

By definition, Min wins if and only if the sequence of the running sums of the weights satisfies the following two conditions:
\begin{itemize}
\item \textbf{\emph{(a)}} infinitely many of its elements are smaller than some constant $C$;
\item \textbf{\emph{(b)}} it has only finitely many $0$'s.
\end{itemize}
Let us show that the condition \textbf{\emph{(b)}} can be replaced by the following one: 
\begin{itemize}
\item \textbf{\emph{(c)}}  $\mathcal{M}_n$ comes into state $0$ only finitely many times on our sequence of weights.
\end{itemize}
For the \textbf{\emph{(b)}}$\implies$\textbf{\emph{(c)}} direction observe that $\mathcal{M}_n$ can be in state $0$ only if indeed the current sum of the weights is $0$. Hence if the sum was $0$ only finitely many times, then $\mathcal{M}_n$ was in state $0$ only finitely many times as well. For the other direction, however, a more subtle argument is needed. This is because the sum of weights can be $0$ while $\mathcal{M}_n$ is in state $\bot$ (this may happen if previously the sum exceeded $n$ in the absolute value). So \textbf{\emph{(c)}}$\implies$\textbf{\emph{(b)}} may be false only if the sum was $0$ infinitely many times after $\mathcal{M}_n$ came into state $\bot$. However, due to our assumptions about $\mathcal{A}$ the sum \emph{never} equals $0$ after $\mathcal{M}_n$ has reached $\bot$. Namely, recall that in $\mathcal{A}$ all cycles are non-negative. Hence (by the same argument as in Case 1) there is no path of weight smaller than $-n$. So $\mathcal{M}_n$ can come into $\bot$ only if the sum exceeded $n$. But if this happened, the sum will never be $0$ again (for that we would need a path of weight smaller than $-n$).

So Min wins if and only if both the conditions \textbf{\emph{(a)}} and \textbf{\emph{(c)}} are satisfied. In terms of the arena $\mathcal{M}_n\times \mathcal{A}$ the condition \textbf{\emph{(c)}} just means that the nodes of the form $(0, v)$ should be visited only finitely many times. Now, to finish the argument it is sufficient to show that in $\mathcal{M}_n\times \mathcal{A}$ some positional strategy $\tau$ of Min wins wherever Min has a winning strategy. Indeed, then by turning $\tau$ into corresponding $\mathcal{M}_n$-strategy $T$ in $\mathcal{A}$ (as in Case 1) we obtain a uniform $\mathcal{M}_n$-strategy equilibrium.

Call a cycle of $\mathcal{M}_n\times \mathcal{A}$ \emph{good} if its weight is $0$ and it contains no node of the form $(0, v)$. Call all the other cycles of $\mathcal{M}_n \times\mathcal{A}$ \emph{bad}.

First, observe that if only bad cycles are reachable from a node of $\mathcal{M}_n\times \mathcal{A}$, then Min looses in this node. Indeed, consider any infinite path $h$ which starts at such a node. There exists a simple cycle $C$ such that $h$ passes through it infinitely many times. This cycle must be bad, so it either contains a node of the form $(0, v)$ or its weight is strictly positive. If the first option holds, then Max wins on $h$ because of visiting a node of the form $(0, v)$ infinitely many times. If the second option holds, then Max wins on $h$ because the sum of the weights goes to $+\infty$. Indeed, any prefix of $h$ which passes through $C$ at least $k$ times has weight at least $k - n$ (each pass through $C$ contributes at least $1$, and the rest of the prefix is a path whose weight is at least $-n$). Since $k$ goes to $+\infty$, so does the sum of the weights.

Second, we claim that Min wins in all nodes from where one can reach a good cycle; moreover, there is a single positional strategy $\tau$ of Min which wins in all these nodes. Denote the set of these nodes of $\mathcal{M}_n\times\mathcal{A}$ by $G$. Since all cycles are non-negative, any good cycle contains a simple good subcycle. Hence in $G$ there exists a set $S$ of disjoint good simple cycles such that any node of $G$ has a path to a cycle from $S$. Consider the following strategy $\tau$. If a node belongs to one of the cycles from $S$, then move along this cycle. Otherwise, move by the shortest path to a cycle from $S$. Clearly, from any node of $G$ the strategy $\tau$ first reaches a cycle from $S$ and then stays on it forever. Since all cycles of $S$ of are zero, this means that the sum of the weights is bounded from above; moreover, these cycles do not contain nodes of the form $(0, v)$, so we will see these nodes only finitely many times.

\subsection{Example with no finite-memory equilibrium}
\label{subsec:two}

Consider the arena from Figure \ref{ex}. We show that if the game starts in the square, then Max has a winning strategy, but no finite-memory one. In particular, this means that $\psi$ is not finite-memory determined.

\begin{figure}[h!]
\centering
  \includegraphics[width=0.5\textwidth]{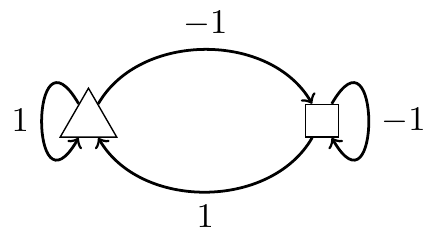}
  \caption{The square is owned by Max and the triangle is owned by Min.}
\label{ex}
\end{figure}

 Namely, the following strategy of Max is winning:  if the current sum of the weights is positive, use a $-1$ loop, otherwise go to the triangle. This strategy guarantees that whenever we reach the square, the sum of the weights will become $0$ once more. So provided that the square is visited infinitely many times, the sum of the weights equals $0$ infinitely often. Now, it might be that from some moment Min stays in the triangle forever, but in this case the sum of the weights goes to $+\infty$, and hence Max also wins.

To show that Max has no finite-memory winning strategy from the square, for every $s$ we define a strategy of Min which wins against any finite-memory strategy of Max with at most $s$ states. Namely, this strategy of Min is as follows: if the current sum of the weights is smaller than $s + 2$, then use a $1$ loop, otherwise go to the square.

Obviously, this strategy guarantees that the sum of the weights never exceeds $s + 2$. Hence Max can win only by making the sum of weights equal to $0$ infinitely many times. However, for that there must be infinitely many periods in which Max stays in the square for at least $s + 1$ moves. But if Max stays in the square for $s + 1$ moves, then it stays there forever after (during these $s + 1$ moves he was in the same state twice). This means that the sum of the weights goes to $-\infty$ and that Min wins.

\section{Proof of Theorem \ref{thm:example}}
\label{sec:example}

Let the set of colors be $C = \{0, 1\}$. Fix a set $T\subseteq \mathbb{Z}^{+}$. Define a payoff function $\vphi\colon\{0, 1\}^\omega\to\{0, 1\}$ by setting $\vphi(\alpha) = 1$ for $\alpha = \alpha_1 \alpha_2 \alpha_3\ldots \in\{0, 1\}^\omega$ if and only if at least one of the following two conditions holds:
\begin{itemize}
\item $\alpha$ contains only finitely many $0$'s;
\item for some $t\in T$ a word
\[0\underbrace{11\ldots 1}_{\mbox{$t$ times}}0\]
is a subword of $\alpha$ (here we call $x\in\{0, 1\}^*$ a subword of $\alpha$ if $x = \alpha_n \alpha_{n + 1} \ldots \alpha_{n + |x| - 1}$ for some $n\in\mathbb{Z}^+$).
\end{itemize}

Call $T$ \emph{sparse} if there are infinitely many $k\in T$ such that $l \notin T$ for all $k < l < k^4$.

\begin{lem}
\label{sparse}
If $T\subseteq \mathbb{Z}^+$ is sparse, then $\vphi\in\oneFMD(f)$ for some $f\colon\mathbb{Z}^+\to\mathbb{Z}^+, f\notin \Omega(n)$.
\end{lem}

Call $T$ \emph{isolated} if for all $m\in\mathbb{Z}^+$ there exists $k\in T, k > m$ such that $l\notin T$ for all $k - m < l < k + m$, $l\neq k$.

\begin{lem}
\label{isolated}
If $T$ is isolated, then $\vphi$ is not arena-independent finite-memory determined.
\end{lem}

Assuming Lemmas \ref{sparse} and \ref{isolated} are proved, it remains to construct a sparse isolated set. For instance, one can take
$T = \{2^{4^n} \mid n\in\mathbb{Z}^+\}$.

\subsection{Proof of Lemma \ref{sparse}}

Call $k\in T$ \emph{good} if $l\notin T$ for all $k < l < k^4$. Define $f$ as follows:
\[f(n) = \min\left\{k + 4 \mid k\in T\mbox{ is good}, k\ge 4 \mbox{ and } k^2 \ge n\right\}\]

By definition of sparseness,  there are infinitely many good $k$ in $T$. Hence, $f(n)$ is well-defined for every $n\in\mathbb{Z}^+$. Now, if $k\in T$ is good and $k\ge 4$, then $f(k^2) \le k + 4$. Since this holds for infinitely many $k$, we have that $f\notin\Omega(n)$.

It remains to show that $\vphi\in\oneFMD(f)$. It is sufficient to construct, for every good $k\in T$, $k\ge 4$, a memory skeleton $\mathcal{M}_k$ with $k + 4$ states such that all one-player arenas with at most $k^2$ nodes have a uniform $\mathcal{M}_k$-strategy equilibrium. In fact, we will show this for all arenas with at most $k^2$ nodes, not only for one-player ones.

Let us define $\mathcal{M}_k$.
States of $\mathcal{M}_k$ will be denoted as follows:
\[I, F, q_0, q_1, \ldots, q_k, q_{>k}.\]
State $I$ is the initial one. Our memory skeleton stays in it until it sees the first $0$. Once the $0$ is seen, $\mathcal{M}_k$ starts memorizing the number of $1$'s after the last $0$ read so far, until this number exceeds $k$. So once we see the first $0$, we come into $q_0$ (there were no $1$'s after this $0$ yet). Next, if we see $1$ in a state $q_i$ for $0 \le i < k$, then we come into $q_{i+1}$. In turn, if we see $1$ in $q_k$, we come into $q_{>k}$ and stay in it as long as we see only $1$'s.

When a new $0$ appears, this interrupts the previous sequence of consecutive $1$'s. Correspondingly, when $\mathcal{M}_k$ sees $0$, in most of the cases it comes into $q_0$. However, in some cases it comes into state $F$ in which it then stays forever. More specifically, if $\mathcal{M}_k$ sees $0$  in a state $q_i$ for $i \le k, i\notin T$, or in $q_{>k}$, then it comes into $q_0$. In turn, if $\mathcal{M}_k$ sees $0$ in a state $q_i$ for $i\le k, i\in T$, then it comes into $F$.

\medskip

Take any arena $\mathcal{A} = \langle V, V_\Max, V_\Min, E, \source, \target, \col\rangle$ with at most $k^2$ nodes. Our goal is to show an existence of a uniform $\mathcal{M}_k$-strategy equilibrium in $\mathcal{A}$ (with respect to $\vphi$). By Observation \ref{prod_obs} it is sufficient to establish an $\{I\}\times V$-wise positional equilibrium in a product arena $\mathcal{M}_k\times\mathcal{A}$ (again, with respect to $\vphi$).

Define an auxiliary payoff $\psi\colon\{0, 1\}^\omega\to\{0, 1\}$ by setting $\psi(\alpha) = 1$ for $\alpha\in\{0, 1\}^\omega$ if and only if either $\alpha$ contains only finitely many $0$'s or a word
\[0\underbrace{11\ldots 1}_{\mbox{$t$ times}}0\]
is a subword of $\alpha$ for some $t\in\{1, 2, \ldots, k\}\cap T$. Our argument consists of proving the following two claims:

\begin{itemize}
\item \emph{Claim 1}. If $(\sigma, \tau)$ is an  $\{I\}\times V$-wise positional equilibrium with respect to $\psi$, then $(\sigma, \tau)$ is also an $\{I\}\times V$-wise positional equilibrium with respect to $\vphi$. Here $\sigma$ is a positional strategy of Max in $\mathcal{M}_k\times\mathcal{A}$ and $\tau$ is a positional strategy of Min in $\mathcal{M}_k\times\mathcal{A}$. 
\item \emph{Claim 2}. There exists an $\{I\}\times V$-wise positional equilibrium in $\mathcal{M}_k\times \mathcal{A}$ with respect to $\psi$.
\end{itemize}

\emph{Proving Claim 1.} It is sufficient to show that as long as $\sigma$ (correspondingly, $\tau$)  is winning in a node $(I, w), w\in V$ with respect to $\psi$, then $\sigma$ (correspondingly, $\tau$) is winning in this node with respect to $\vphi$. For $\sigma$ this is immediate because $\psi(\alpha) = 1 \implies \vphi(\alpha) = 1$ for every $\alpha\in\{0, 1\}^\omega$. Now, let $(I, w)$ be a node for which $\tau$ is winning with respect to $\psi$. Assume for contradiction that there exists an infinite path with the source in $(I, w)$ which is consistent with $\tau$ and which is winning for Max with respect to $\vphi$.
Since this path is loosing for Max with respect to $\psi$,  the corresponding infinite sequence of colors must have  a subword of the form
\[0\underbrace{11\ldots 1}_{\mbox{$l$ times}}0\]
for some $l\in T\setminus \{1, 2, \ldots, k\}$. Since $k$ is good, we must have $l\ge k^4$. This means that in  $(\mathcal{M}_k \times \mathcal{A})_\tau$ there is a path which starts in $(I, w)$ and contains $k^4$ consecutive edges colored by $1$.
Now, there are at mot $(k + 4) \cdot k^2$ nodes in $\mathcal{M}_k\times\mathcal{A}$. Since $k\ge 4$, we have $(k + 4) \cdot k^2 < k^4$. This means that in $(\mathcal{M}_k \times \mathcal{A})_\tau$ one can reach from $(I, w)$ a cycle colored only by $1$'s. Therefore there is a strategy of Max such that in its play against $\tau$ there are only finitely many edges colored by $0$. Hence $\tau$ could not be winning in $(I, w)$ with respect to $\psi$, contradiction. 

\medskip

\emph{Proving Claim 2.} We will show that there is a parity game with 3 priorities on $\mathcal{M}_k\times\mathcal{A}$ such that for every $w\in V$ Max wins in this parity game from $(I, w)$ if and only if he wins from $(I, w)$ with respect to $\psi$. Once this claim is proved it remains to refer to the positional determinacy of parity games.

It is easy to see that $\mathcal{M}_k$ is in state $F$ if and only if the current sequence of colors has a subword
\[0\underbrace{11\ldots 1}_i0\]
for some $i\in T\cap\{1, 2, \ldots, k\}$. So Max wins with respect to $\psi$ in the following two cases: \textbf{\emph{(a)}} $\mathcal{M}_k$ ever comes into $F$  \textbf{\emph{(b)}} all but finitely many edges of a play are colored by $1$. In terms of the arena $\mathcal{M}_k\times \mathcal{A}$ condition \textbf{\emph{(a)}} means that a node of the form $(F, v)$ is visited at least once. To put it differently, \textbf{\emph{(a)}} means that some \emph{edge} which starts in a node of the form $(F, v)$ is passed at least once. Let us denote the set of these edges by $E_2$. Partition all the other edges of $\mathcal{M}_k\times\mathcal{A}$ into two sets $E_0$ and $E_1$ according to their color. So Max wins if either a play contains an edge from $E_2$ or it contains only finitely many edges of $E_0$. 
 In fact, instead of requiring to have at least one edge from $E_2$ we may require to have infinitely many such edges (because once $\mathcal{M}_k$ came into $F$, it stays in it forever). This is equivalent to a parity game with $3$ priorities. Namely, label edges from $E_2$ by priority $3$, edges from $E_0$ by priority $2$, and edges from $E_1$ by priority $1$.  Observe that Max wins if and only if the largest priority visited infinitely many times is odd.

\subsection{Proof of lemma \ref{isolated}}

For every $m$ we construct a one-player arena $\mathcal{A}_m$ for which there exists no memory skeleton $\mathcal{M}$ with less than $m$ states such that $\mathcal{A}_m$ has a uniform $\mathcal{M}$-strategy equilibrium. Clearly, this implies that $\vphi$ is not arena-independent finite-memory determined.

By definition of isolation, there exists $k\in T, k > m$ such that $l\notin T$ for all $k - m < l < k + m$, $l\neq k$. Let $\mathcal{A}_m$ be as on Figure \ref{is}. All its nodes are owned by Max. On the left it has $m$ nodes. The $i$th one (from the top) has a single simple path to the central node; the colors along this simple path form a word $01^i$ (a zero followed by $i$ ones). On the right $\mathcal{A}_m$ also has $m$ nodes.  For each $i\in \{1, 2, \ldots, m\}$ there is a single simple path from the central node to the $i$th node (from the top) on the right; colors along this path form a word $1^{k - i}$. Finally, all the nodes on the right have a $0$ loop.

\begin{figure}[h!]
\centering
  \includegraphics[width=0.5\textwidth]{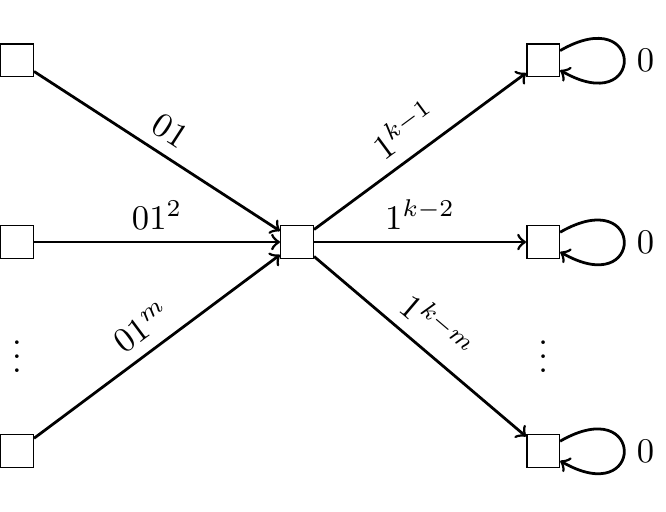}
  \caption{Arena $\mathcal{A}_m$.}
\label{is}
\end{figure}

Assume for contradiction that there exists a memory skeleton $\mathcal{M}$ with less than $m$ states such that $\mathcal{A}_m$ has a uniform $\mathcal{M}$-strategy equilibrium.This means that some $\mathcal{M}$-strategy $\Sigma$ of Max is winning for him in all the nodes where Max has a winning strategy. Observe that all the nodes on the left are winning for Max. Indeed, from the $i$th one Max should go  (through the central node) to the $i$th node on the right. The resulting infinite sequence of colors will be  $01^k 0^\omega$, and this is winning for Max since $k\in T$. Note also that this is the only infinite path which is winning for Max from the $i$th node on the left. Indeed, any other path is colored by  $01^{i + k - j} 0^\omega$ for some $j\in\{1, 2, \ldots, m\}, j\neq i$. This is loosing for Max because  $i + k - j \neq k$ , $k - m < i + k - j  < k + m$, and hence $ i + k - j \notin T$.

 Our $\mathcal{M}$-strategy $\Sigma$ must be winning for all the nodes on the left. So if the game starts in the $i$th node on the left, then in the central node $\Sigma$ must go to the $i$th node on the right. However, as there are less than $m$ states in $\mathcal{M}$, there must be two distinct nodes on the left from which $\mathcal{M}$ comes into the same state upon reaching the central node. So $\mathcal{M}$ must make the same move from the central node no matter in which of these two nodes on the left the game started. This means that $\Sigma$ will be loosing for at least one of these two nodes.


\appendix
\section{Proof of Observation \ref{prod_obs}}
\label{app:prod_obs}
Let $(\widehat{\Sigma}, \widehat{T})$ be an  $(\{m_{init}\}\times S)$-wise positional equilibrium in $\mathcal{M}\times\mathcal{A}$ (here $\widehat{\Sigma}$ is a strategy of Max and $\widehat{T}$ is a strategy of Min). Define an $\mathcal{M}$-strategy $\Sigma$ of Max in $\mathcal{A}$ as follows. To determine $\Sigma(h)$ for a position $h$ with $\target(h) = v\in V_\Max$ and $\delta(m_{init}, \col(h)) = m$, we consider a move which $\widehat{\Sigma}$ makes in a node $(m, v)$. Assume that this move is a pair $(m, e)$ (we must have $\source(e) = v$). Then we set $\Sigma(h) = e$. We define an $\mathcal{M}$-strategy $T$ of Min in $\mathcal{A}$ similarly through the strategy $\widehat{T}$. We claim that $(\Sigma, T)$ is an $S$-wise equilibrium in $\mathcal{A}$.

Assume for contradiction that for some $v\in S$ either $\Sigma$ is not an optimal response to $T$ or $T$ is not an optimal response to $\Sigma$ in $v$. We consider only the first option, the second one can be treated similarly. Then some infinite path $h \in \Cons(v, T)$ is better from the Max's perspective than $h(v, \Sigma, T)$ (the play of $\Sigma$ and $T$ from $v$), i.e., 
\begin{equation}
\label{h}
\vphi \circ \col(h) > \vphi \circ \col\big(h(v, \Sigma, T)\big).
\end{equation}
 For $n\in\mathbb{Z}^+$ let $e_n^\prime$ denote the $n$th edge of $h$ and $e_n$ denote the $n$th edge of $h(v, \Sigma, T)$.
For each of these two sequences of edges define a sequence of states into which $\mathcal{M}$ comes while reading colors of these edges (assuming $\mathcal{M}$ is initially in $m_{init}$):
\begin{align}
\label{m_prime}
m_1^\prime &= m_{init}, \qquad m_{n + 1}^\prime = \delta(m_n^\prime, \col(e_n^\prime)) \mbox{ for every $n\in\mathbb{Z}^+$},\\
\label{m}
m_1 &= m_{init}, \qquad m_{n + 1} = \delta(m_n, \col(e_n)) \mbox{ for every $n\in\mathbb{Z}^+$}.
\end{align}
It is easy to see that the sequence
$(m_1, e_1)(m_2, e_2) (m_3, e_3)\ldots$
is the play of $\widehat{\Sigma}$ and $\widehat{T}$ from $(m_{init}, v)$. For example, let us show its consistency with $\widehat{\Sigma}$. We have to show that for every $n\in\mathbb{Z}^+$ such that $\source((m_n, e_n)) = (m_n, \source(e_n))$ is a node of Max we have $\widehat{\Sigma}((m_n, \source(e_n))) = (m_n, e_n)$. By definition of $\Sigma$ it is sufficient to show that $e_{n} = \Sigma(h)$ for a position $h$ in $\mathcal{A}$ with $\target(h) = \source(e_n)$ and $\delta(m_{init}, \col(h)) = m_{n}$. It is easy to see that we have this for a position $h = e_1 \ldots e_{n - 1}$ if $n > 1$ and for $h = \lambda_v$ if $n = 1$. Indeed, we have $\Sigma(h) = e_n$ and $\target(h) = \source(e_n)$ because $e_1 e_2\ldots e_n$ is a prefix of the play of $\Sigma$ and $T$ from $v$. Now, we have $\delta(m_{init}, \col(h)) = m_n$ because of \eqref{m}.

Consistency of $(m_1, e_1)(m_2, e_2) (m_3, e_3)\ldots$ with $\widehat{T}$ can be shown similarly. Moreover, by the same argument  the sequence
$(m_1^\prime, e_1^\prime)(m_2^\prime, e_2^\prime) (m_3^\prime, e_3^\prime)\ldots$
is also consistent with $\widehat{T}$, due to \eqref{m_prime}. Now, since $\widehat{\Sigma}$ is an optimal response to $\widehat{T}$ in $(m_{init}, v)$, we have that $(m_1, e_1)(m_2, e_2) (m_3, e_3)\ldots$ is at least as good as $(m_1^\prime, e_1^\prime)(m_2^\prime, e_2^\prime) (m_3^\prime, e_3^\prime)\ldots$ from the Max's perspective. However, by definition the first sequence is colored exactly as $h(v, \Sigma, T)$ and the second one exactly as $h$. This is a contradiction with \eqref{h}.
\end{document}